\documentclass[12pt]{amsart}
\usepackage{amssymb}
\usepackage{color}
\usepackage{amsmath,epic,curves,amscd}
\usepackage[english]{babel}
\usepackage{graphicx}
\usepackage{comment}
\usepackage{appendix}
\usepackage{mathdots}
\usepackage[all]{xy}
\usepackage{relsize}
\DeclareMathSymbol{\shortminus}{\mathbin}{AMSa}{"39}

\newtheorem{claim}{}[section]
\newtheorem{theorem}[claim]{Theorem}

\newtheorem{corollary}[claim]{Corollary}

\theoremstyle{remark}

\renewenvironment{proof}{\noindent{\it Proof. \hskip0pt}}
                      {$\square$\par\medskip}

\begin{document}
\baselineskip 6.0 truemm
\parindent 1.5 true pc

\newcommand\xx{{\text{\sf X}}}
\newcommand\lan{\langle}
\newcommand\ran{\rangle}
\newcommand\tr{{\text{\rm Tr}}\,}
\newcommand\ot{\otimes}
\newcommand\ol{\overline}
\newcommand\join{\vee}
\newcommand\meet{\wedge}
\renewcommand\ker{{\text{\rm Ker}}\,}
\newcommand\image{{\text{\rm Im}}\,}
\newcommand\id{{\text{\rm id}}}
\newcommand\tp{{\text{\rm tp}}}
\newcommand\pr{\prime}
\newcommand\e{\epsilon}
\newcommand\la{\lambda}
\newcommand\inte{{\text{\rm int}}\,}
\newcommand\ttt{{\text{\rm t}}}
\newcommand\spa{{\text{\rm span}}\,}
\newcommand\conv{{\text{\rm conv}}\,}
\newcommand\rank{\ {\text{\rm rank of}}\ }
\newcommand\re{{\text{\rm Re}}\,}
\newcommand\ppt{\mathbb T}
\newcommand\rk{{\text{\rm rank}}\,}
\newcommand\SN{{\text{\rm SN}}\,}
\newcommand\SR{{\text{\rm SR}}\,}
\newcommand\HA{{\mathcal H}_A}
\newcommand\HB{{\mathcal H}_B}
\newcommand\HC{{\mathcal H}_C}
\newcommand\CI{{\mathcal I}}
\newcommand{\bra}[1]{\langle{#1}|}
\newcommand{\ket}[1]{|{#1}\rangle}
\newcommand\cl{\mathcal}
\newcommand\idd{{\text{\rm id}}}
\newcommand\OMAX{{\text{\rm OMAX}}}
\newcommand\OMIN{{\text{\rm OMIN}}}
\newcommand\diag{{\text{\rm Diag}}\,}
\newcommand\calI{{\mathcal I}}
\newcommand\bfi{{\bf i}}
\newcommand\bfj{{\bf j}}
\newcommand\bfk{{\bf k}}
\newcommand\bfl{{\bf l}}
\newcommand\bfp{{\bf p}}
\newcommand\bfq{{\bf q}}
\newcommand\bfzero{{\bf 0}}
\newcommand\bfone{{\bf 1}}
\newcommand\sa{{\rm sa}}
\newcommand\ph{{\rm ph}}
\newcommand\phase{{\rm ph}}
\newcommand\res{{\text{\rm res}}}
\newcommand{\algname}[1]{{\sc #1}}
\newcommand{\Setminus}{\setminus\hskip-0.2truecm\setminus\,}
\newcommand\calv{{\mathcal V}}
\newcommand\calg{{\mathcal G}}
\newcommand\calt{{\mathcal T}}
\newcommand\calvnR{{\mathcal V}_n^{\mathbb R}}
\newcommand\D{{\mathcal D}}
\newcommand\C{{\mathcal C}}
\newcommand\aaa{\alpha}
\newcommand\bbb{\beta}
\newcommand\ccc{\gamma}
\newcommand\xxxx{\par\bigskip {\color{red}========================================} \bigskip\par}
\newcommand\tefrac{\textstyle\frac}
\newcommand\xxx{\xx^\sigma}
\newcommand\E{{\mathcal E}}
\newcommand\W{{\mathcal W}}
\newcommand\X{{\mathcal X}}
\newcommand\variable{{\mathcal X}}
\newcommand\LX{{\mathcal L}_{\text{\sf X}}}
\newcommand\ext{{\rm Ext}\,}
\newcommand\GHZ{{\sf GHZ_d}}
\newcommand\pp{\phantom{-}}
\newcommand\CE{{\mathcal C}{\mathcal E}}

\newcommand{\exta}{{\mathcal E}^1}
\newcommand{\extb}{{\mathcal E}^2}
\newcommand{\extc}{{\mathcal E}^3}

\title{Criteria for partial entanglement of three qubit states arising from distributive rules}

\author{Kyung Hoon Han and Seung-Hyeok Kye}
\address{Kyung Hoon Han, Department of Data Science, The University of Suwon, Gyeonggi-do 445-743, Korea}
\email{kyunghoon.han at gmail.com}
\address{Seung-Hyeok Kye, Department of Mathematics and Institute of Mathematics, Seoul National University, Seoul 151-742, Korea}
\email{kye at snu.ac.kr}
\thanks{Both KHH and SHK were partially supported by NRF-2020R1A2C1A01004587, Korea}



\begin{abstract}
It is known that the partial entanglement/separability violates distributive rules with respect to
the operations of taking convex hull and intersection. In this note, we give criteria for
three qubit partially entangled states arising from distributive rules,
together with the corresponding witnesses.
The criteria will be given in terms of diagonal and anti-diagonal entries. They
actually characterize those partial entanglement completely
when all the entries are zero except for diagonal and anti-diagonal entries.
Important states like Greenberger-Horne-Zeilinger diagonal states fall down in this class.
\end{abstract}

\maketitle

\section{Introduction}

The notion of entanglement from quantum physics is now one of the
most important resources in current quantum information and quantum
computation theory. Recall that a state is called (fully) separable
if it is a convex sum of pure product states, and entangled if it is
not separable. In multi-partite systems, the notion of entanglement
depends on the partitions of subsystems to get various kinds of
partial entanglement. In the tri-partite system with subsystems $A$,
$B$ and $C$, we have three kinds of bi-separability, that is,
$A$-$BC$, $B$-$CA$ and $C$-$AB$ separability according to the
bi-partitions of the subsystems. We call those {\sl basic
bi-separability}. We will denote by $\aaa$, $\bbb$ and $\ccc$ the
convex cones consisting of all $A$-$BC$, $B$-$CA$ and $C$-$AB$
separable un-normalized states, respectively. In the three qubit
cases, the convex cones $\aaa,\bbb$ and $\ccc$ are sitting in the
real vector space of all three qubit self-adjoint matrices, which is a $64$ dimensional real vector space.

Many authors have considered the intersections and convex hulls for basic bi-separable states, which will be denoted by
$\meet$ and $\join$, respectively. Note that the convex hull of two convex cones coincides with the nonnegative sum.
After it was shown in \cite{bdmsst} that a three qubit state in $\aaa\meet\bbb\meet\ccc$
need not to be fully separable as a tri-partite state, several authors have considered intersections and convex sums of the convex cones
$\aaa$, $\bbb$ and $\ccc$. See \cite{dur-cirac,dur-cirac-tarrach} for intersections
of two of them, and \cite{seevinck-uffink,acin} for convex hull of them.
See also \cite{sz2011,sz2012,sz2015,sz2018,han_kye_bisep_exam,han_kye_pe,Szalay-2019} for further development
in more general contexts.
We recall that a tri-partite state is called genuinely entangled if it does not belong to
$\aaa\join\bbb\join\ccc$.

Very recently, the authors and Szalay \cite{han_kye_szalay} considered the lattice, dented by ${\mathcal L}$, generated by
three convex cones $\aaa,\bbb$ and $\ccc$ with respect to the two operations of convex hull and
intersection, and showed that this lattice violates the distributive rules. More precisely,
it was shown that the following inequalities
\begin{align}
(\aaa\meet\bbb)\join(\aaa\meet\ccc)&\le \aaa\meet(\bbb\join\ccc),\label{dist-ineq}\\
\aaa\join(\bbb\meet\ccc)&\le (\aaa\join\bbb)\meet(\aaa\join\ccc)\label{dist-ineq-2}
\end{align}
are strict.
We refer to \cite{birkhoff,{freese}} for general theory of lattices.

For a convex cone $C$ in a real vector space $V$ with a bi-linear pairing $\lan\, \cdot\, ,\, \cdot \, \ran$, the dual cone $C^\circ$ is defined by
the convex cone consisting of all $y\in V$ satisfying $\lan x,y\ran\ge 0$ for every $x\in C$.
For given two self-adjoint matrices $x=[x_{ij}]$ and $y=[y_{ij}]$, we use the bi-linear pairing
$$
\lan x,y\ran=\tr(xy^\ttt)=\sum_{ij}x_{ij}y_{ij},
$$
where $y^\ttt$ denotes the transpose of $y$.
We recall that matrices in the dual cones play the roles of {\sl witnesses}.
For example, we have $\varrho\notin (\aaa\meet\bbb)\join(\aaa\meet\ccc)$ if and only if
there exists
$$
W\in [(\aaa\meet\bbb)\join(\aaa\meet\ccc)]^\circ
=(\aaa^\circ\join\bbb^\circ)\meet(\aaa^\circ\join\ccc^\circ)
$$
such that $\lan W,\varrho\ran<0$. See \cite{han_kye_pe} for the details.
We note that the following
inequalities
\begin{align}
\aaa^\circ\join(\bbb^\circ\meet\ccc^\circ)&\le (\aaa^\circ\join\bbb^\circ)\meet(\aaa^\circ\join\ccc^\circ),\label{dist-ineq-2-dual}\\
(\aaa^\circ\meet\bbb^\circ)\join(\aaa^\circ\meet\ccc^\circ)&\le \aaa^\circ\meet(\bbb^\circ\join\ccc^\circ)\label{dist-ineq-dual}
\end{align}
are also strict, by duality.

The main purpose of this paper is to give criteria for the convex cones arising in the above inequalities
(\ref{dist-ineq}), (\ref{dist-ineq-2}), (\ref{dist-ineq-2-dual}) and (\ref{dist-ineq-dual}) in the three qubit cases.
Criteria will be given in terms of diagonal and anti-diagonal entries.
Criteria of such kinds have been considered for the convex cone $\aaa\join\bbb\join\ccc$ in
\cite{{gao},{guhne10},{Rafsanjani}} to get sufficient conditions for genuine entanglement.
Such criteria for $\aaa,\bbb,\ccc$ and $\aaa^\circ,\bbb^\circ,\ccc^\circ$ also can be found in
\cite[Proposition 5.2]{han_kye_optimal} and \cite[Theorem 6.2]{han_kye_tri}, respectively,
(see also Propositions 3.1 and 3.3 of \cite{han_kye_optimal}) from which
we also get criteria for intersections like $\aaa\meet\bbb$ and $\aaa^\circ\meet\bbb^\circ$.
Convex sums like $\aaa\join\bbb$ and $\aaa^\circ\join\bbb^\circ$ have been considered in \cite{han_kye_pe}.
Finally, we found criteria for the convex cones of the type $\aaa\join(\bbb\meet\ccc)$
in \cite{han_kye_szalay} in the context of distributive rules. Therefore,
we will concentrate on the convex cones of the following types
\begin{equation}\label{list}
(\aaa\meet\bbb)\join(\aaa\meet\ccc),\qquad
\aaa^\circ\join(\bbb^\circ\meet\ccc^\circ),\qquad
(\aaa^\circ\meet\bbb^\circ)\join(\aaa^\circ\meet\ccc^\circ).
\end{equation}
The whole convex cones of partially separable states we are considering can be drawn in the following diagram with the inclusion relations.
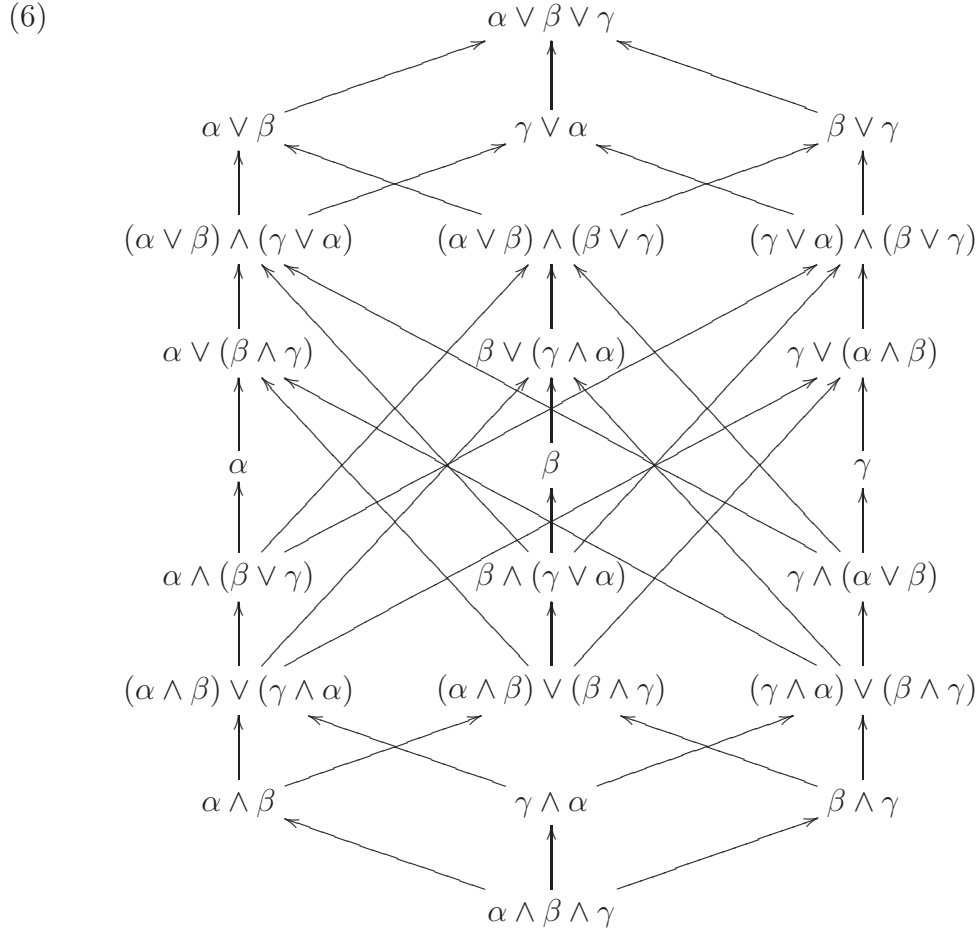
\begin{equation}\label{diag_con}
\xymatrix{
& \alpha \join \beta \join \gamma && \\
\alpha \join \beta \ar[ur] & \gamma \join \alpha \ar[u] & \beta \join \gamma \ar[ul] & \\
(\alpha \join \beta) \meet (\gamma \join \alpha) \ar[u] \ar[ur]
    & (\alpha \join \beta) \meet (\beta \join \gamma) \ar[ul] \ar[ur]
    & (\gamma \join \alpha) \meet (\beta \join \gamma) \ar[ul] \ar[u] & \\
\alpha \join (\beta \meet \gamma) \ar[u] & \beta \join (\gamma \meet \alpha) \ar[u] & \gamma \join (\alpha \meet \beta) \ar[u]
&
\\
\alpha \ar[u] & \beta \ar[u] & \gamma \ar[u] &  \\
\alpha \meet (\beta \join \gamma) \ar[u] \ar@{>}[uuurr] \ar@{>}[uuur] 
     &\beta \meet (\gamma \join \alpha) \ar[u] \ar@{>}[uuul] \ar@{>}[uuur] 
     & \gamma \meet (\alpha \join \beta) \ar[u] \ar@{>}[uuull] \ar@{>}[uuul] 
     &
\\
(\alpha \meet \beta) \join (\gamma \meet \alpha) \ar[u] \ar@{>}[uuur] \ar@{>}[uuurr] 
    & (\alpha \meet \beta) \join (\beta \meet \gamma) \ar[u] \ar@{>}[uuul] \ar@{>}[uuur] 
    & (\gamma \meet \alpha) \join (\beta \meet \gamma) \ar[u] \ar@{>}[uuull] \ar@{>}[uuul] 
    & \\
\alpha \meet \beta \ar[u] \ar[ur] & \gamma \meet \alpha \ar[ul] \ar[ur] & \beta \meet \gamma \ar[ul] \ar[u] & \\
& \alpha \meet \beta \meet \gamma \ar[ul] \ar[u] \ar[ur] &
}
\end{equation}
We also have the similar diagram for the dual cones consisting of witnesses.

We recall that a matrix is called $\xx$-{\sl shaped} \cite{yu07} if all the entries are zero except for diagonal and anti-diagonal entries.
$\xx$-shaped states will be called $\xx$-states. Many important states like GHZ diagonal states belong to this class.
If a three qubit state is (partially) separable, then its $\xx$-part is also (partially) separable,
and so any separability criteria for \xx-states
will give rise to a necessary criteria for general three qubit states. In this
paper, we will give necessary conditions for three qubit states
(respectively, witnesses) to belong to convex cones
listed in (\ref{list}) (respectively, the dual of (\ref{list})) in
terms of diagonal and anti-diagonal entries. Such conditions are
also sufficient when states/witnesses are $\xx$-shaped. Especially,
we will find lattice identities for the first and third convex cones
in (\ref{list}) when states/witnesses are $\xx$-shaped. See
Corollaries \ref{identity_P} and \ref{identity_Q}.

After we collect known results for criteria in the next section, we
will give criteria for the convex cones of the type
$(\aaa\meet\bbb)\join(\aaa\meet\ccc)$ in Section 3, where we will
also show that the lattice ${\mathcal L}$ is not complemented.
Criteria for the types $\aaa^\circ\join(\bbb^\circ\meet\ccc^\circ)$
and $(\aaa^\circ\meet\bbb^\circ)\join(\aaa^\circ\meet\ccc^\circ)$
will be given in Sections 4 and 5, respectively.
In Section 6, we restrict our attention on GHZ diagonal states,
to exhibit all GHZ diagonal states belonging to convex cones considered in the paper.

The authors are grateful to the referee for bringing their attention to GHZ diagonal states.

\section{summary of the known criteria}

States and witnesses in the tensor product $M_2\otimes M_2\otimes M_2$ of $2\times 2$ matrices may be written
as an $8\times 8$ matrices with respect to the lexicographic order of of indices for subsystems. Then $\xx$-shaped states/witnesses
are of the form
$$
\xx(a,b,z)= \left(
\begin{matrix}
a_1 &&&&&&& z_1\\
& a_2 &&&&& z_2 & \\
&& a_3 &&& z_3 &&\\
&&& a_4&z_4 &&&\\
&&& \bar z_4& b_4&&&\\
&& \bar z_3 &&& b_3 &&\\
& \bar z_2 &&&&& b_2 &\\
\bar z_1 &&&&&&& b_1
\end{matrix}
\right),
$$
for $a,b\in\mathbb R^4$ and $z\in\mathbb C^4$. It is well-known that
every GHZ diagonal state \cite{GHZ} is in this form, and an $\xx$-state  $\varrho=\xx(a,b,z)$
is GHZ diagonal if and only if $a=b$ and $z\in\mathbb R^4$.
See \cite{han_kye_GHZ}.

We denote by $\varrho_\xx$ the $\xx$-part of a state $\varrho$, in the obvious sense. If we denote by $\mathcal L$
the lattice generated by $\aaa$, $\bbb$ and $\ccc$ in the three qubit case, then we have
\begin{equation}\label{x-part}
\varrho\in\sigma \ \Longrightarrow \varrho_\xx\in\sigma,
\end{equation}
for every $\sigma\in\mathcal L$. In fact, it is easily seen that if (\ref{x-part}) holds for $\sigma$ and $\tau$ then
it also holds for $\sigma\meet\tau$ and $\sigma\join\tau$. We have already seen that
the generators $\aaa,\bbb$ and $\ccc$ of ${\mathcal L}$ satisfy (\ref{x-part})
in \cite[Proposition 2.2]{han_kye_pe}, and so it follows that (\ref{x-part}) holds for every $\sigma\in\mathcal L$.
If we denote by ${\mathcal L}^\circ$ the lattice generated by $\aaa^\circ$, $\bbb^\circ$ and $\ccc^\circ$, then we also have
\begin{equation}\label{x-part-w}
W\in\sigma \ \Longrightarrow W_\xx\in\sigma,
\end{equation}
for every $\sigma\in{\mathcal L}^\circ$, by the identity $\lan\varrho_\xx,W\ran = \lan\varrho, W_\xx\ran$.

By a {\sl pair} $\{i,j\}$, we always mean throughout this paper an unordered set
with {\sl two} elements among $1,2,3,4$, that is, we assume that $i\neq j$.
For a given three qubit $\xx$-shaped state $\varrho=\xx(a,b,z)$, we
consider the inequalities
\begin{center}
\framebox{
\parbox[t][2.5cm]{14.80cm}{
\addvspace{0.1cm} \centering
$$
\begin{array}{ll}
S_1[i,j]:&\quad \min\{\sqrt{a_ib_i},\sqrt{a_jb_j}\}\ge \max\{|z_i|,|z_j|\},\\
S_2[i,j]:&\quad
\min\left\{\sqrt{a_ib_i}+\sqrt{a_jb_j},\sqrt{a_k b_k}+\sqrt{a_\ell b_\ell}\right\}
   \ge\max\left\{|z_i|+|z_j|,|z_k|+|z_\ell|\right\},\\
S_3 :&\quad \textstyle{\sum_{j\neq i}\sqrt{a_jb_j}\ge |z_i|},\quad i=1,2,3,4,\\
S_4[i,j|k,\ell]:&\quad
\min\left\{\sqrt{a_ib_i}+\sqrt{a_jb_j},\sqrt{a_k b_k}+\sqrt{a_\ell b_\ell}\right\}
   \ge\max\left\{|z_i|+|z_j|,|z_k|+|z_\ell|\right\}.
\end{array}
$$
}}
\end{center}\medskip
The inequalities $S_1$ and $S_2$ are defined for a pair $\{i,j\}$, where the pair
$\{k,\ell\}$ appearing in the inequality $S_2[i,j]$ are chosen so that $\{i,j,k,\ell\}=\{1,2,3,4\}$.
On the other hand, the inequality $S_4$ is defined for arbitrary two pairs $\{i,j\}$ and $\{k,\ell\}$.
We note that the inequality $S_4[i,j|k,\ell]$ holds
automatically for any {\sf X}-states $\varrho=\xx(a,b,z)$ when $\{i,j\}=\{k,\ell\}$.
If $\{i,j\}\cap\{k,\ell\}=\emptyset$, we note that
the three inequalities $S_2[i,j]$, $S_2[k,\ell]$ and $S_4[i,j|k,\ell]$ are same.

We summarize the results for three qubit $\xx$-shaped states $\varrho=\xx(a,b,z)$ as follows:
First of all, it was shown in \cite[Proposition 5.2]{han_kye_optimal} and
\cite[Proposition 3.1]{han_kye_pe} that
\begin{itemize}
\item
$\varrho\in\aaa$ if and only if $S_1[1,4]$ and $S_1[2,3]$ hold;
\item
$\varrho\in\bbb$ if and only if $S_1[1,3]$ and $S_1[2,4]$ hold;
\item
$\varrho\in\ccc$ if and only if $S_1[1,2]$ and $S_1[3,4]$ hold.
\end{itemize}
Analogous results for multi-qubit states are also known in \cite{han_kye_optimal}.
As for the convex hulls of them, the authors showed in
\cite[Theorem 5.5]{han_kye_pe} the following:
\begin{itemize}
\item
$\varrho\in\bbb\join\ccc$ if and only if $S_2[1,4]$ (equivalently $S_2[2,3]$) holds;
\item
$\varrho\in\ccc\join\aaa$ if and only if $S_2[1,3]$ (equivalently $S_2[2,4]$) holds;
\item
$\varrho\in\aaa\join\bbb$ if and only if $S_2[1,2]$ (equivalently $S_2[3,4]$) holds.
\end{itemize}
On the other hand, it has been known earlier \cite{{guhne10},{gao},{Rafsanjani},{han_kye_optimal}}
that
\begin{itemize}
\item
$\varrho\in\aaa\join\bbb\join\ccc$
if and only if $S_3$ holds.
\end{itemize}
See also \cite[Proposition 4.5]{han_kye_pe}.
Finally, the authors and Szalay showed
in \cite[Theorem 2.1]{han_kye_szalay} that the following
\begin{itemize}
\item
$\varrho\in\aaa\join(\bbb\meet\ccc)$ if and only if $S_4[i,j|k,\ell]$ holds
whenever $\{i,j\}$, $\{k,\ell\}$ are two of $\{1,2\}$, $\{1,3\}$, $\{2,4\}$, $\{3,4\}$;
\item
$\varrho\in\bbb\join(\ccc\meet\aaa)$ if and only if $S_4[i,j|k,\ell]$ holds
whenever $\{i,j\}$, $\{k,\ell\}$ are two of $\{1,2\}$, $\{1,4\}$, $\{2,3\}$, $\{3,4\}$;
\item
$\varrho\in\ccc\join(\aaa\meet\bbb)$ if and only if $S_4[i,j|k,\ell]$ holds
whenever $\{i,j\}$, $\{k,\ell\}$ are two of $\{1,3\}$, $\{1,4\}$, $\{2,3\}$, $\{2,4\}$
\end{itemize}
hold for $\xx$-states $\varrho=\xx(a,b,z)$ in the contexts of distributive rules.

For an {\sf X}-shaped
self-adjoint matrix $W=\xx(s,t,u)$ with $s_i,t_i\ge 0$ and
$u\in\mathbb C^4$, we also consider the following inequalities:
\begin{center}
\framebox{
\parbox[t][1.9cm]{10.00cm}{
\addvspace{0.1cm} \centering
$$
\begin{array}{ll}
W_1[i,j]: &\quad \sqrt{s_it_i}+\sqrt{s_jt_j}\ge |u_i|+|u_j|,\\
W_2[i,j]: &\quad \sum_{k\neq j}\sqrt{s_kt_k}\ge |u_i|,\quad \sum_{k\neq i}\sqrt{s_kt_k}\ge |u_j|,\qquad \qquad \qquad \qquad\qquad\\
W_3 :&\quad \sum_{i=1}^4\sqrt{s_it_i}\ge\sum_{i=1}^4|u_i|
\end{array}
$$
}}
\end{center}\medskip
for a pair $\{i,j\}$.
Three qubit $\xx$-shaped witnesses $W=\xx(s,t,u)$ which are dual of the basic bi-separability
have been considered in \cite[Theorem 6.2]{han_kye_tri} and
\cite[Proposition 3.3]{han_kye_pe} as follows:
\begin{itemize}
\item
$W\in\aaa^\circ$ if and only if $W_1[1,4]$ and $W_1[2,3]$ hold;
\item
$W\in\bbb^\circ$ if and only if $W_1[1,3]$ and $W_1[2,4]$ hold;
\item
$W\in\ccc^\circ$ if and only if $W_1[1,2]$ and $W_1[3,4]$ hold.
\end{itemize}
On the other hand, the joins of them have been characterized in
\cite[Theorem 5.2]{han_kye_pe}:
\begin{itemize}
\item
$W\in\bbb^\circ\join\ccc^\circ$ if and only if $W_2[1,4]$, $W_2[2,3]$ and $W_3$ hold;
\item
$W\in\ccc^\circ\join\aaa^\circ$ if and only if $W_2[1,3]$, $W_2[2,4]$ and $W_3$ hold;
\item
$W\in\aaa^\circ\join\bbb^\circ$ if and only if $W_2[1,2]$, $W_2[3,4]$ and $W_3$ hold.
\end{itemize}
Finally, it was shown in
\cite[Theorem 5.5]{han_kye_optimal} and
\cite[Proposition 4.2]{han_kye_pe} that
\begin{itemize}
\item
$W\in\aaa^\circ\join\bbb^\circ\join\ccc^\circ$  if and only if $W_3$ holds.
\end{itemize}

\section{Criteria for $(\aaa\meet\bbb)\join(\aaa\meet\ccc)$}

In this section, we look for criteria for the convex cones of the type $(\aaa\meet\bbb)\join(\aaa\meet\ccc)$. To do this,
we first consider the following inequalities;
\begin{equation}\label{dist1}
\begin{aligned}
(x\meet y)\join (x\meet z)
&\le x\meet (x\join (y\meet z))\meet(y\join (z\meet x))\meet (z\join (y\meet x))\\
&= x\meet (y\join (z\meet x))\meet (z\join (y\meet x))\\
&\le x\meet (y\join (z\meet x))\\
&\le x\meet (y\join z),
\end{aligned}
\end{equation}
which hold in general in an arbitrary lattice.
From the first inequality of (\ref{dist1}),
we have natural necessary conditions. We show that they are sufficient
for $\xx$-states.

\begin{theorem}\label{dist111}
For a three qubit state $\varrho$ with the $\xx$-part $\xx(a,b,z)$, we have the following:
\begin{enumerate}
\item[(i)]
if $\varrho\in (\aaa\meet\bbb)\join(\aaa\meet\ccc)$, then the
inequalities $S_1[1,4]$, $S_1[2,3]$ and $S_4[i,j|k,\ell]$ hold for
any different pairs $\{i,j\}$ and $\{k,\ell\}$;
\item[(ii)]
if $\varrho\in (\bbb\meet\ccc)\join(\bbb\meet\aaa)$, then the
inequalities $S_1[1,3]$, $S_1[2,4]$ and $S_4[i,j|k,\ell]$ hold for
any different pairs $\{i,j\}$ and $\{k,\ell\}$;
\item[(iii)]
if $\varrho\in (\ccc\meet\aaa)\join(\ccc\meet\beta)$, then the
inequalities $S_1[1,2]$, $S_1[3,4]$ and $S_4[i,j|k,\ell]$ hold for
any different pairs $\{i,j\}$ and $\{k,\ell\}$.
\end{enumerate}
If $\varrho=\xx(a,b,z)$ then the converses also hold.
\end{theorem}

\begin{proof}
We prove the last one. The necessity follows from (\ref{dist1}), more precisely from the inequality
$$
(\ccc\meet\aaa)\join(\ccc\meet\beta)\le
\ccc\meet (\aaa\join(\bbb\meet\ccc))\meet(\bbb\join(\ccc\meet\aaa))\meet (\ccc\join(\aaa\meet\bbb)),
$$
and the criteria \cite[Theorem 2.1]{han_kye_szalay} for convex cones in the right side.
We prove the converse
when $\varrho=\xx(a,b,z)$. To do this, we suppose that $\varrho$ satisfies $S_1[1,2]$, $S_1[3,4]$
and $S_4[i,j|k,\ell]$ for different pairs $\{i,j\}$ and $\{k,\ell\}$. In order to show that $\varrho\in
(\ccc\meet\aaa)\join(\ccc\meet\bbb)$, we suppose
that $W=\xx(s,t,u)\in (\ccc^\circ\join\aaa^\circ)\meet(\ccc^\circ\join\bbb^\circ)$ and
will prove $\lan W,\varrho\ran\ge 0$. So, we assume that $W$ satisfies $W_2[1,3]$,
$W_2[2,4]$ $W_2[1,4]$, $W_2[2,3]$, and $W_3$.

If $\varrho\in\aaa$ or $\varrho\in\bbb$, then there is nothing to prove since $\varrho\in\ccc$ by
$S_1[1,2]$ and $S_1[3,4]$. Therefore, we may assume that
$\varrho\notin\aaa$ and $\varrho\notin\bbb$. By $\varrho\notin\aaa$, we may assume $|z_4|>\sqrt{a_1b_1}$ without loss of generality.
By $\varrho\notin\bbb$, we have one of the following:
$$
|z_3|>\sqrt{a_1b_1} ,\qquad |z_1|>\sqrt{a_3b_3}, \qquad
|z_2|>\sqrt{a_4b_4}, \qquad  |z_4|>\sqrt{a_2b_2}.
$$
The second implies $|z_1|+|z_4|>\sqrt{a_1b_1}+\sqrt{a_3b_3}$, which violates $S_4[1,4|1,3]$.
The third also violates $S_4[2,4|1,4]$. Therefore, we have two cases:
 Either $|z_4|, |z_3| \ge \sqrt{a_1b_1}$ or $|z_4| \ge \sqrt{a_1b_1}, \sqrt{a_2b_2}$.

We consider the case: $|z_4|, |z_3| \ge \sqrt{a_1b_1}$.
By $S_4[1,i|3,4]$ with $i=2,3,4$, we have
$$
\min\{ \sqrt{a_2b_2}, \sqrt{a_3b_3}, \sqrt{a_4b_4} \} \ge |z_4| + (|z_3|-\sqrt{a_1b_1}).
$$
Therefore, we have
$$
\begin{aligned}
&\sqrt{s_2t_2}\sqrt{a_2b_2} + \sqrt{s_3t_3}\sqrt{a_3b_3} + \sqrt{s_4t_4}\sqrt{a_4b_4} -|u_4||z_4| \\
&\phantom{XX} \ge (\sqrt{s_2t_2} + \sqrt{s_3t_3} + \sqrt{s_4t_4}) \min \{ \sqrt{a_2b_2}, \sqrt{a_3b_3}, \sqrt{a_4b_4}) -|u_4||z_4|\\
&\phantom{XX} \ge  \left( \sqrt{s_2t_2} + \sqrt{s_3t_3} + \sqrt{s_4t_4} -|u_4| \right) |z_4| +
     (\sqrt{s_2t_2}+\sqrt{s_3t_3}+\sqrt{s_4t_4})(|z_3| -\sqrt{a_1b_1}) \\
&\phantom{XX} \ge  \left( \sqrt{s_2t_2} + \sqrt{s_3t_3} + \sqrt{s_4t_4} -|u_4| \right) \sqrt{a_1b_1}
     + (\sqrt{s_2t_2}+\sqrt{s_3t_3}+\sqrt{s_4t_4})(|z_3| -\sqrt{a_1b_1}) \\
&\phantom{XX}=-|u_4|\sqrt{a_1b_1}+(\sqrt{s_2t_2}+\sqrt{s_3t_3}+\sqrt{s_4t_4})|z_3|,
\end{aligned}
$$
where the last inequality follows from $W_2[4,1]$.
On the other hand, we also have
$$
\sqrt{s_1t_1}\sqrt{a_1b_1}-|u_1||z_1|-|u_2||z_2|-|u_3||z_3|
\ge
\sqrt{s_1t_1}\sqrt{a_1b_1}-|u_1|\sqrt{a_1b_1}-|u_2|\sqrt{a_1b_1}-|u_3||z_3|,
$$
by $S_1[1,2]$. Summing up the above two inequalities, we have
$$
\begin{aligned}
{\tfrac 12}\lan W, \varrho \ran
& = {1 \over 2}\textstyle\sum_{i=1}^4 [s_i a_i + t_i b_i + 2 {\rm Re} (u_i z_i)]\\
& \ge \textstyle\sum_{i=1}^4 (\sqrt{s_i t_i} \sqrt{a_i b_i} - |u_i| |z_i|)\\
& \ge(\sqrt{s_1t_1}-|u_1|-|u_2|-|u_4|)\sqrt{a_1b_1}
  +(\sqrt{s_2t_2}+\sqrt{s_3t_3}+\sqrt{s_4t_4}-|u_3|)|z_3|\\
&\ge (\sqrt{s_1t_1}-|u_1|-|u_2|-|u_4|)\sqrt{a_1b_1}
  +(\sqrt{s_2t_2}+\sqrt{s_3t_3}+\sqrt{s_4t_4}-|u_3|)\sqrt{a_1b_1}\\
&=(\textstyle{\sum_{i=1}^4\sqrt{s_it_i}-\sum_{j=1}^4|u_j|})\sqrt{a_1b_1},
\end{aligned}
$$
by $W_2[3,1]$. This is nonnegative by $W_3$.

It remains to consider the case: $|z_4| \ge \sqrt{a_1b_1}, \sqrt{a_2b_2}$.
We use $S_4[1,2|i,4]$ with $i=1,2,3$, to get
$\sqrt{a_1b_1} + \sqrt{a_2b_2} - |z_4| \ge  \max\{ |z_1|, |z_2|, |z_3| \}$, and so
$$
\begin{aligned}
(\textstyle\sum_{i=1}^4\sqrt{s_it_i}-|u_4|)(\sqrt{a_1b_1} + \sqrt{a_2b_2} - |z_4| )
&\ge  (\textstyle\sum_{i=1}^3|u_i|) \max\{ |z_1|, |z_2|, |z_3| \}\\
&\ge \textstyle\sum_{i=1}^3|u_i||z_i|,
\end{aligned}
$$
by $W_3$. Therefore, we have
$$
\begin{aligned}
 \textstyle\sum_{i=1}^4 &(\sqrt{s_i t_i} \sqrt{a_i b_i} - |u_i| |z_i|) \\
& = \textstyle\sum_{i=1}^2\sqrt{s_it_i}\sqrt{a_ib_i}
    + \left(\sqrt{s_3t_3}\sqrt{a_3b_3} + \sqrt{s_4t_4}\sqrt{a_4b_4} -|u_4||z_4|\right) - \textstyle\sum_{i=1}^3 |u_i||z_i| \\
& \ge  \textstyle\sum_{i=1}^2\sqrt{s_it_i}\sqrt{a_ib_i}+ (\sqrt{s_3t_3} + \sqrt{s_4t_4} -|u_4|)|z_4| \\
& \phantom{XXXXXXXXXXX} - (\textstyle\sum_{i=1}^4 \sqrt{s_it_i} - |u_4| ) (\sqrt{a_1b_1} + \sqrt{a_2b_2} - |z_4|),
\end{aligned}
$$
by $S_1[3,4]$. We continue as follows:
$$
\begin{aligned}
& =  ( - \textstyle\sum_{i \ne 1} \sqrt{s_it_i} + |u_4| ) \sqrt{a_1b_1} + ( - \textstyle\sum_{i \ne 2} \sqrt{s_it_i} + |u_4| ) \sqrt{a_2b_2}\\
& \phantom{XXXXXXXXXXX} + ( \sqrt{s_1t_1} + \sqrt{s_2t_2} + 2\sqrt{s_3t_3} + 2\sqrt{s_4t_4} - 2|u_4| )|z_4| \\
& \ge  ( - \textstyle\sum_{i \ne 1} \sqrt{s_it_i} + |u_4| ) |z_4|
+ ( - \textstyle\sum_{i \ne 2} \sqrt{s_it_i} + |u_4| ) |z_4|\\
& \phantom{XXXXXXXXXXX} + ( \sqrt{s_1t_1} + \sqrt{s_2t_2} + 2\sqrt{s_3t_3} + 2\sqrt{s_4t_4} - 2|u_4| )|z_4|=0, \\
\end{aligned}
$$
by the inequalities $W_2[4,1]$ and $W_2[4,2]$.
\end{proof}

In order to summarize the role of inequality $S_4[i,j|k,\ell]$, we consider the following six convex cones
$$
\begin{aligned}
\sigma_1&=\aaa\join\bbb,\\
\sigma_2&=(\aaa\join\bbb)\meet(\aaa\join\ccc),\\
\sigma_3&=\aaa\join(\bbb\meet\ccc),\\
\sigma_4&=\aaa,\\
\sigma_5&=\aaa\meet(\bbb\join\ccc),\\
\sigma_6&=(\aaa\meet\bbb)\join(\aaa\meet\ccc),
\end{aligned}
$$
which make the chain $\sigma_1\supset\sigma_2\supset\sigma_3\supset\sigma_4\supset\sigma_5\supset\sigma_6$
of inclusions. We provide Table 1 to see which inequalities $S_4[i,j|k,\ell]$ we need to determine $\sigma_i$ for $i=1,2,3,4,5,6$.
We recall again that both $S_2[i,j]$ and $S_2[k,\ell]$
coincide with $S_4[i,j|k,\ell]$ when $\{i,j\}\cap\{k,\ell\}=\emptyset$.

\begin{table}
\begin{center}
\begin{tabular}{c||c|c|c|c|c|c|}
&$\{1,2\}$ &$\{1,3\}$ &$\{1,4\}$ &$\{2,3\}$ &$\{2,4\}$ &$\{3,4\}$\\
\hline\hline
$\{1,2\}$ && $\sigma_3$ &$\sigma_6$ &$\sigma_6$  &$\sigma_3$ &$\sigma_1$\\
\hline
$\{1,3\}$ &  &&$\sigma_6$ &$\sigma_6$  &$\sigma_2$ &$\sigma_3$\\
\hline
$\{1,4\}$ && &&$\sigma_5$  &$\sigma_6$ &$\sigma_6$\\
\hline
$\{2,3\}$ &&& &&$\sigma_6$ &$\sigma_6$\\
\hline
$\{2,4\}$ &&&& &&$\sigma_3$\\
\hline
$\{3,4\}$ &&&&& &\\
\hline
\end{tabular}
\caption{\small This table shows that $\sigma_1=\aaa\join\bbb$ is determined by $S_4[1,2|3,4]$.
The convex cone $\sigma_2=(\aaa\join\bbb)\meet(\aaa\join\ccc)$
is determined by an extra inequality $S_4[1,3|2,4]$, and $\sigma_3=\aaa\join(\bbb\meet\ccc)$ is determined by the six inequalities
labeled by $\sigma_1$, $\sigma_2$ and $\sigma_3$.
The convex cone $\sigma_4=\aaa$ is determined by another inequalities $S_1[1,4]$ and $S_1[2,3]$ which imply
these six inequalities. One more inequality $S_4[1,4|2,3]$ is required in order to determine $\sigma_5=\aaa\meet(\bbb\join\ccc)$,
and all the inequalities are required to determine $\sigma_6=(\aaa\meet\bbb)\join(\aaa\meet\ccc)$ as well as $S_1[1,4]$ and $S_1[2,3]$.
We note that $S_1[i,j]$ may be recovered if we allow $S_4[i,i|j,j]$.
}
\end{center}
\end{table}

We denote by $\xx$ the real vector space of all three qubit $\xx$-shaped self-adjoint matrices, and define
$$
\sigma_X:=\sigma\cap\xx=\{\varrho_\xx:\varrho\in\sigma\},
$$
for each $\sigma\in{\mathcal L}$, where the last identity follows from (\ref{x-part}).
 Then, it is easily seen that
\begin{equation}\label{homox}
(\sigma\meet\tau)_\xx=\sigma_\xx\meet\tau_\xx,\qquad
(\sigma\join\tau)_\xx=\sigma_\xx\join\tau_\xx,
\end{equation}
and so, we see that
$$
{\mathcal L}_\xx=\{ \sigma_\xx:\sigma\in{\mathcal L} \}
$$
is the lattice generated by $\aaa_\xx$, $\bbb_\xx$ and $\ccc_\xx$, which are
$16$ affine dimensional convex bodies sitting in the real vector space $\xx$.
We write
$$
P:=[\aaa\join(\bbb\meet\ccc)]\meet[\bbb\join(\ccc\meet\aaa)]\meet[\ccc\join(\aaa\meet\bbb)]
$$
in the lattice ${\mathcal L}$. Then we have
$$
P_\xx =[\aaa_\xx\join(\bbb_\xx\meet\ccc_\xx)]\meet[\bbb_\xx\join(\ccc_\xx\meet\aaa_\xx)]\meet[\ccc_\xx\join(\aaa_\xx\meet\bbb_\xx)]
$$
in ${\mathcal L}_\xx$.
Theorem \ref{dist111} shows that the following lattice theoretic identities hold among generators
$\aaa_\xx$, $\bbb_\xx$ and $\ccc_\xx$ of the lattice ${\mathcal L}_\xx$.

\begin{corollary}\label{identity_P}
We have the following identities
$$
\begin{aligned}
(\aaa_\xx\meet\bbb_\xx)\join(\aaa_\xx\meet\ccc_\xx)&=\aaa_\xx\meet P_\xx,\\
(\bbb_\xx\meet\ccc_\xx)\join(\bbb_\xx\meet\aaa_\xx)&=\bbb_\xx\meet P_\xx,\\
(\ccc_\xx\meet\aaa_\xx)\join(\ccc_\xx\meet\bbb_\xx)&=\ccc_\xx\meet P_\xx.
\end{aligned}
$$
\end{corollary}

We also write
$$
Q:=[\aaa\meet(\bbb\join\ccc)]\join[\bbb\meet(\ccc\join\aaa)]\join[\ccc\meet(\aaa\join\bbb)]
$$
in the lattice ${\mathcal L}$. We will also have later the dual identities
for $\aaa_\xx\join Q_\xx$, $\bbb_\xx\join Q_\xx$ and $\ccc_\xx\join Q_\xx$
in the lattice ${\mathcal L}_\xx$.

It was asked in \cite{han_kye_szalay} whether the lattice ${\mathcal L}$ is complemented or not.
Recall that a lattice $L$ is called complemented if every $x\in L$ has a {\sl complement} $y\in L$
satisfying $x\meet y=0$ and $x\join y=1$, where $0$ and $1$ are the least and greatest elements of $L$,
respectively. Note that $\aaa\meet\bbb\meet\ccc$ and $\aaa\join\bbb\join\ccc$
are the least and the greatest elements of the lattice ${\mathcal L}$.
We will show that $\aaa$ has no complement in the lattice ${\mathcal L}$. To do this, we recall
the results in \cite{han_kye_pe}. We denote by $\Delta$ the collection of eight diagonal states
$\xx(E_i,0,0)$ and $\xx(0,E_i,0)$ with the usual orthonormal basis $\{E_1,E_2,E_3,E_4\}$,
which generate extreme rays of every convex cones in ${\mathcal L}$
by \cite[Theorem 4.3]{han_kye_pe}.
We also denote by $\ext (C)$
the set of points of a convex cone $C$ which generate extreme rays, and put
$$
\E_\aaa=\ext(\aaa_X)\setminus \Delta,\qquad
\E_\bbb=\ext(\bbb_X)\setminus \Delta,\qquad
\E_\ccc=\ext(\ccc_X)\setminus \Delta.
$$
All the states in $\E_\aaa$, $\E_\bbb$, $\E_\ccc$ and $\ext(\aaa_\xx\join\bbb_\xx\join\ccc_\xx)$
have been found in Theorem 3.5 and Theorem 4.6 of \cite{han_kye_pe}. Especially, we have the following:
\begin{itemize}
\item
$\E_\aaa$, $\E_\bbb$ and $\E_\ccc$ are mutually disjoint;
\item
$\ext(\aaa_\xx\join\bbb_\xx\join\ccc_\xx)$ coincides with the disjoint union $\E_\aaa\sqcup \E_\bbb\sqcup\E_\ccc\sqcup\Delta$.
\end{itemize}
Now, we assume that $\aaa$ has a complement $\sigma$ in the lattice ${\mathcal L}$. Then we have
$$
\aaa_\xx\join\bbb_\xx\join\ccc_\xx =\aaa_\xx\join\sigma_\xx,\qquad
\aaa_\xx\meet\bbb_\xx\meet\ccc_\xx =\aaa_\xx\meet\sigma_\xx,
$$
by (\ref{homox}). Then we have
$$
\E_\bbb\cup\E_\ccc
\subset\ext(\aaa_\xx\join\bbb_\xx\join\ccc_\xx)
=\ext(\aaa_\xx\join\sigma_\xx)
\subset\ext(\aaa_\xx)\cup\ext(\sigma_\xx),
$$
which implies $\E_\bbb\cup\E_\ccc \subset\ext(\sigma_\xx)$ by the mutual disjointness of
$\E_\aaa$, $\E_\bbb$, $\E_\ccc$ and $\Delta$. We also have
$$
\Delta\subset\aaa_\xx\meet\bbb_\xx\meet\ccc_\xx =\aaa_\xx\meet\sigma_\xx\subset\sigma_\xx,
$$
and so we have $\bbb_\xx\join\ccc_\xx\subset\sigma_\xx$. Therefore, we have
$$
\aaa_\xx\meet\bbb_\xx\meet\ccc_\xx\subsetneqq
\aaa_\xx\meet(\bbb_\xx\join\ccc_\xx)\subset \aaa_\xx\meet\sigma_\xx=\aaa_\xx\meet\bbb_\xx\meet\ccc_\xx
$$
by the criteria in Section 2. This contradiction shows that the lattice ${\mathcal L}$ is not complemented.

Returning to the inequalities in (\ref{dist1}), we recall that a lattice $L$ is called distributive if all the inequalities in (\ref{dist1})
are identities for every $x,y,z\in L$, and modular if $(x\meet y)\join (x\meet z)$ coincides with
$x\meet (y\join (z\meet x))$ for every $x,y,z\in L$. We have considered in \cite{han_kye_szalay} the $\xx$-state given by
$$
\varrho_1=\xx\left((2,1,1,2)(2,1,1,2),(2,0,1,0)\right)\in\aaa\meet(\bbb\join(\ccc\meet\aaa))
$$
in order to show that the lattice ${\mathcal L}$ is not modular. We also consider
$$
\varrho_2=\xx\left( (2,1,1,2),(2,1,1,2), (2,1,0,0)\right)\in\aaa\meet(\ccc\join(\bbb\meet\aaa))
$$
to see that the strict inequalities
$$
\begin{aligned}
(\aaa_\xx\meet\bbb_\xx)\join(\aaa_\xx\meet\ccc_\xx)&\lneqq \aaa_\xx\meet(\bbb_\xx\join(\ccc_\xx\meet\aaa_\xx)),\\
(\aaa_\xx\meet\bbb_\xx)\join(\aaa_\xx\meet\ccc_\xx)&\lneqq \aaa_\xx\meet(\ccc_\xx\join(\bbb_\xx\meet\aaa_\xx))
\end{aligned}
$$
hold. If we take meet of these two formulae, we see that the identity holds
when we plug $\aaa_\xx$, $\bbb_\xx$ and $\ccc_\xx$ into $x,y$ and $z$
in following general inequality
\begin{equation}\label{aaaaa}
(x\meet y)\join (x\meet z) \le x\meet (y\join (z\meet x))\meet (z\join (y\meet x)).
\end{equation}
See also Corollary \ref{identity_Q}
for the dual identity.
These identities are very special from the view point of general lattice theory or convex geometry.
To see this, we consider the lattice of all convex sets on the plain with respect to the convex hull and intersection.
In this lattice, we take a closed disc with a diameter $\overline{AB}$ and two line segments
$\overline{AC}$ and $\overline{BD}$ so that these line segments touch the disc
at single points $A$ and $B$, respectively. We plug the disc, the line segments $\overline{AC}$ and $\overline{BD}$
into $x,y$ and $z$ in (\ref{aaaaa}). Then the left side is the just line segment $\overline{AB}$, but the right
side is the intersection of the two triangles $\triangle ABC$ and $\triangle ABD$ inside of the disc,
which is much bigger that the line segment $\overline{AB}$ in general.

\section{Criteria for $\aaa^\circ\join(\bbb^\circ\meet\ccc^\circ)$}

In this section, we give criteria for the convex cones of the type  $\aaa^\circ\join(\bbb^\circ\meet\ccc^\circ)$.
For a self-adjoint $W=\xx(s,t,u)$, we consider the inequality $W_4[i,j]$ which combines the following two inequalities:
\begin{center}
\framebox{
\parbox[t][1.4cm]{10.00cm}{
\addvspace{0.1cm} \centering
$$
\begin{array}{ll}
W_{4{\rm a}}[i,j]: &\quad \sqrt{s_it_i}+\sqrt{s_j t_j}+2\min\{\sqrt{s_k t_k},\sqrt{s_\ell t_\ell}\} \ge |u_i|+|u_j|,\\
W_{4{\rm b}}[i,j]: &\quad \sqrt{s_it_i}+\sqrt{s_jt_j}+2(\sqrt{s_kt_k}+\sqrt{s_\ell t_\ell})\ge |u_i|+|u_j|+2\max\{|u_k|,|u_\ell|\}\\
\end{array}
$$
}}
\end{center}\medskip
for a pair $\{i,j\}$, where $\{k,\ell\}$ is chosen so that $\{i,j,k,\ell\}=\{1,2,3,4\}$. We have the following:

\begin{theorem}\label{join-statestep4}
If a self-adjoint three qubit matrix $W$ with the $\xx$-part $\xx(s,t,u)$ belongs to the convex cone
$$
\aaa^\circ\join(\bbb^\circ\meet\ccc^\circ),\qquad
{\rm(}{\text respectively}\
\bbb^\circ\join(\ccc^\circ\meet\aaa^\circ)\
{\text and}\
\ccc^\circ\join(\aaa^\circ\meet\bbb^\circ)
{\rm )}
$$
then $W$ satisfies $W_3$ together with the following:
\begin{enumerate}
\item[(i)]
$W_2[i,j]$ whenever $\{i,j\}$ is one of $\{1,2\}$, $\{1,3\}$, $\{2,4\}$, $\{3,4\}$
{\rm (}respectively
$\{1,2\}$, $\{1,4\}$, $\{2,3\}$, $\{3,4\}$ and $\{1,3\}$, $\{1,4\}$, $\{2,3\}$, $\{2,4\}${\rm )};
\item[(ii)]
$W_{4}[i,j]$ whenever
$\{i,j\}$ is one of $\{1,4\}$, $\{2,3\}$ {\rm (}respectively
$\{1,3\}$, $\{2,4\}$ and $\{1,2\}$, $\{3,4\}${\rm )}.
\end{enumerate}
If $W$ is $\xx$-shaped, then the converse holds.
\end{theorem}

\begin{proof}
We will prove for the convex cone $\ccc^\circ\join(\aaa^\circ\meet\bbb^\circ)$.
Suppose that $W\in \ccc^\circ\join(\aaa^\circ\meet\bbb^\circ)$. The required inequalities $W_2[i,j]$ and $W_3$
follow from
$$
\ccc^\circ\join(\aaa^\circ\meet\bbb^\circ)\le (\ccc^\circ\join\aaa^\circ)\meet (\ccc^\circ\join\bbb^\circ).
$$
To get the inequalities $W_{4{\rm a}}$ and $W_{4{\rm b}}$, we
may assume that $s_i,t_i>0$ as in the proof of \cite[Theorem 2.1]{han_kye_szalay}.
We consider
$$
\begin{aligned}
\varrho_{i,j,k}:=&\xx(\sqrt{\tfrac{t_i}{s_i}}E_i + \sqrt{\tfrac{t_j}{s_j}}E_j + 2\sqrt{\tfrac{t_k}{s_k}}E_k,\\
                  &\qquad\qquad\qquad \sqrt{\tfrac{s_i}{t_i}}E_i + \sqrt{\tfrac{s_j}{t_j}}E_j + 2\sqrt{\tfrac{s_k}{t_k}}E_k,\\
                 &\qquad\qquad\qquad\qquad\qquad\qquad -e^{-{\rm i}\theta_i}E_i -e^{-{\rm i}\theta_j}E_j), \\
\varrho_{i,j,k,\ell}':=&\xx(\sqrt{\tfrac{t_i}{s_i}}E_i + \sqrt{\tfrac{t_j}{s_j}}E_j + 2\sqrt{\tfrac{t_k}{s_k}}E_k
                           + 2\sqrt{\tfrac{t_\ell}{s_\ell}}E_\ell,\\
                  &\qquad\qquad\qquad \sqrt{\tfrac{s_i}{t_i}}E_i + \sqrt{\tfrac{s_j}{t_j}}E_j + 2\sqrt{\tfrac{s_k}{t_k}}E_k
                             + 2\sqrt{\tfrac{s_\ell}{t_\ell}}E_\ell, \\
                   &\qquad\qquad\qquad\qquad\qquad\qquad -e^{-{\rm i}\theta_i}E_i - e^{-{\rm i}\theta_j}E_j - 2e^{-{\rm i}\theta_k}E_k),
\end{aligned}
$$
where $\theta_m = \arg z_m$.
When $\{\{i,j\},\{k,\ell\}\}=\{\{1,2\},\{3,4\}\}$, both of them satisfy $S_1[1,2]$, $S_1[3,4]$, $S_2[1,2]$, and so,
belong to $\ccc\meet(\aaa\join\bbb)$.
We expand $\langle W, \varrho_{i,j,k} \rangle \ge 0$ and $\langle W, \varrho_{i,j,k,\ell}' \rangle \ge 0$
to get the required inequalities $W_{\rm 4a}[i,j]$ and $W_{\rm 4b}[i,j]$.

For the converse, it suffices to show the
inequality $\lan W,\varrho\ran\ge 0$ under the following assumptions:
\begin{itemize}
\item
$W=\xx(s,t,u)$ satisfies $W_3$;
\item
$W=\xx(s,t,u)$ satisfies
$W_2[i,j]$ whenever $\{i,j\}$ is one of $\{1,3\}$, $\{1,4\}$, $\{2,3\}$, $\{2,4\}$;
\item
$W=\xx(s,t,u)$ satisfies
both $W_{4{\rm a}}[i,j]$ and $W_{4{\rm b}}[i,j]$
whenever
$\{i,j\}$ is one of $\{1,2\}$, $\{3,4\}$;
\item
$\varrho=\xx(a,b,z)\in\ccc\meet(\aaa\join\bbb)$, or equivalently
$\varrho$ satisfies
$S_1[1,2]$, $S_1[3,4]$ and $S_2[1,2]$.
\end{itemize}
By the inequalities $W_3$, $W_2[1,4]$ and $W_2[2,3]$, we have $W\in\bbb^\circ\join\ccc^\circ$.
Similarly, we also have $W\in \ccc^\circ\join\aaa^\circ$ by $W_3$, $W_2[1,3]$ and $W_2[2,4]$.
If $\varrho\in (\ccc\meet\aaa)\join(\ccc\meet\bbb)$, then we have $\lan W,\varrho\ran\ge 0$ by the duality.
If $\varrho\notin (\ccc\meet\aaa)\join(\ccc\meet\bbb)$ then both $\varrho\notin\aaa$ and $\varrho\notin\bbb$ hold, since $\varrho\in\ccc$.
By $\varrho\notin\aaa$, we may assume that
$$
|z_4|>\sqrt{a_1b_1}
$$
without loss of generality. As for $\varrho\notin\bbb$, we have the following four possibilities:
$$
|z_3|>\sqrt{a_1b_1}, \qquad  |z_1|>\sqrt{a_3b_3}, \qquad |z_2|>\sqrt{a_4b_4}, \qquad  |z_4|>\sqrt{a_2b_2}.
$$
The second implies $|z_4|>\sqrt{a_1b_1} \ge |z_1| > \sqrt{a_3b_3}\ge |z_4|$ by $S_1[3,4]$, which is a contradiction.
Because the third also implies $|z_2|>\sqrt{a_4b_4} \ge |z_4| > \sqrt{a_1b_1}\ge |z_2|$ by $S_1[1,2]$, we have
two possibilities, the first and the fourth. We consider the following four cases:
\begin{enumerate}
\item[(I)]
$|z_4| \ge |z_3| \ge \sqrt{a_1b_1}$,
\item[(II)]
$|z_3| \ge |z_4| \ge \sqrt{a_1b_1}$,
\item[(III)]
$|z_4| \ge \sqrt{a_1b_1} \ge \sqrt{a_2b_2}$,
\item[(IV)]
$|z_4| \ge \sqrt{a_2b_2} \ge \sqrt{a_1b_1}$.
\end{enumerate}

For the case (I), we use the inequality $S_2[1,2]$ to see
$\sqrt{a_2b_2} \ge |z_3|+|z_4|-\sqrt{a_1b_1}$. Therefore, we have
\begin{align*}
\tfrac 12\lan W,\varrho\ran
&\ge \textstyle\sum_{i=1}^4 \sqrt{s_i t_i} \sqrt{a_i b_i} - |u_i| |z_i|\\
& = ( \sqrt{s_2t_2}\sqrt{a_2b_2} + \sqrt{s_3t_3}\sqrt{a_3b_3} + \sqrt{s_4t_4}\sqrt{a_4b_4} -|u_4||z_4|) \\
& \qquad\qquad\qquad +( \sqrt{s_1t_1}\sqrt{a_1b_1} -|u_1||z_1| -|u_2||z_2| -|u_3||z_3|) \\
& \ge \sqrt{s_2t_2}(|z_3|+|z_4|-\sqrt{a_1b_1}) + \sqrt{s_3t_3}|z_4| + \sqrt{s_4t_4}|z_4| - |u_4||z_4| \\
& \qquad\qquad\qquad  +( \sqrt{s_1t_1}\sqrt{a_1b_1} -|u_1|\sqrt{a_1b_1} -|u_2|\sqrt{a_1b_1} -|u_3||z_3|),
\end{align*}
by $S_1[3,4]$ and $S_1[1,2]$. We continue as follows:
\begin{align*}
 & =(\sqrt{s_2t_2} + \sqrt{s_3t_3} + \sqrt{s_4t_4} -|u_4|) |z_4| + (\sqrt{s_2t_2}-|u_3|)|z_3|\\
& \qquad\qquad\qquad +(\sqrt{s_1t_1} - \sqrt{s_2t_2} -|u_1| -|u_2|)\sqrt{a_1b_1} \\
& \ge (\sqrt{s_2t_2} + \sqrt{s_3t_3} + \sqrt{s_4t_4} -|u_4|) |z_3| + (\sqrt{s_2t_2}-|u_3|)|z_3| \\
& \qquad\qquad\qquad +(\sqrt{s_1t_1} - \sqrt{s_2t_2} -|u_1| -|u_2|)\sqrt{a_1b_1}
\end{align*}
by $W_2[4,1]$. This is equal to
\begin{align*}
& = (2\sqrt{s_2t_2} + \sqrt{s_3t_3} + \sqrt{s_4t_4} - |u_3| - |u_4|) |z_3|\\
& \qquad\qquad\qquad+(\sqrt{s_1t_1} - \sqrt{s_2t_2} -|u_1| -|u_2|)\sqrt{a_1b_1} \\
& \ge (2\sqrt{s_2t_2} + \sqrt{s_3t_3} + \sqrt{s_4t_4} - |u_3| - |u_4|) \sqrt{a_1b_1}\\
& \qquad\qquad\qquad +(\sqrt{s_1t_1} - \sqrt{s_2t_2} -|u_1| -|u_2|)\sqrt{a_1b_1} \\
& = \left( \textstyle\sum_{i=1}^4 \sqrt{s_it_i} - \textstyle\sum_{i=1}^4 |u_i| \right) \sqrt{a_1b_1}
\end{align*}
by $W_{4{\rm a}}[3,4]$, which is nonnegative by $W_3$.
For the case (II), note that the conclusion and all the conditions on $W$ and $\varrho$ are invariant
under switching the first and the second subsystems, except $|z_3| \ge |z_4| \ge \sqrt{a_1b_1}$.
It changes $|z_3| \ge |z_4| \ge \sqrt{a_1b_1}$
into $|\bar z_4| \ge |\bar z_3| \ge \sqrt{a_1b_1}$, which is exactly the case (I).

For the case (III), we first note the following inequality
$$
|u_3||z_3|\le
\left(\textstyle\sum_{i=1}^4 \sqrt{s_it_i}-|u_1|-|u_2|- |u_4| \right) (\sqrt{a_1b_1}+\sqrt{a_2b_2}-|z_4|)
$$
by $W_3$ and $S_2[1,2]$. Therefore, we have
\begin{align*}
\textstyle\sum_{i=1}^4 \sqrt{s_i t_i} \sqrt{a_i b_i}& - |u_i| |z_i|
\ge ( \sqrt{s_1t_1}\sqrt{a_1b_1} + \sqrt{s_2t_2}\sqrt{a_2b_2} + \sqrt{s_3t_3}|z_4| + \sqrt{s_4t_4}|z_4|) \\
& \qquad +(  -|u_1|\sqrt{a_2b_2} -|u_2|\sqrt{a_2b_2} -|u_4||z_4|) \\
& \qquad -(\textstyle\sum_{i=1}^4 \sqrt{s_it_i}-|u_1|-|u_2|- |u_4| ) (\sqrt{a_1b_1}+\sqrt{a_2b_2}-|z_4|)
\end{align*}
by
$S_1[3,4]$ and $S_1[1,2]$.
This is equal to the following:
\begin{align*}
& = (\sqrt{s_1t_1} + \sqrt{s_2t_2} + 2\sqrt{s_3t_3} + 2\sqrt{s_4t_4} -|u_1| -|u_2| -2|u_4|) |z_4| \\
& \qquad + (-\sqrt{s_2t_2}-\sqrt{s_3t_3}-\sqrt{s_4t_4}+|u_1|+|u_2|+|u_4|)\sqrt{a_1b_1} \\
& \qquad +(-\sqrt{s_1t_1}-\sqrt{s_3t_3}-\sqrt{s_4t_4}+|u_4|)\sqrt{a_2b_2}.
\end{align*}
Using $W_{4{\rm b}}[1,2]$, we may replace $|z_4|$ by $\sqrt{a_1b_1}$ to the smaller quanity
$$
(\sqrt{s_1t_1}+\sqrt{s_3t_3}+\sqrt{s_4t_4}-|u_4|)(\sqrt{a_1b_1}-\sqrt{a_2b_2}).
$$
This is nonnegative by $W_2[4,2]$, and so we completed the proof for the case (III).

It remains to prove the case (IV).
The conclusion and all the conditions on $W$ and $\varrho$ except $|z_4| \ge \sqrt{a_2b_2} \ge \sqrt{a_1b_1}$
are invariant under switching the first and the second subsystems and
the local unitary operation by $I \otimes I \otimes \begin{pmatrix}0&1\\1&0\end{pmatrix}$. They change
$|z_4| \ge \sqrt{a_2b_2} \ge \sqrt{a_1b_1}$ into
$|z_3| = |\bar z_3| \ge \sqrt{a_2b_2} \ge \sqrt{a_1b_1}$, and again into
$|z_4| \ge \sqrt{a_1b_1} \ge \sqrt{a_2b_2}$.
This is the case (III).
\end{proof}

\section{Criteria for $(\aaa^\circ\meet\bbb^\circ)\join(\aaa^\circ\meet\ccc^\circ)$}

As for the convex cones of the type $(\aaa^\circ\meet\bbb^\circ)\join(\aaa^\circ\meet\ccc^\circ)$,
we have the following criteria:

\begin{theorem}\label{dist11122}
For a three qubit self-adjoint matrix $W$ with the $\xx$-part $\xx(s,t,u)$, we have the following:
\begin{enumerate}
\item[(i)]
if $W\in (\aaa^\circ\meet\bbb^\circ)\join(\aaa^\circ\meet\ccc^\circ)$, then inequalities
$W_1[1,4]$, $W_1[2,3]$, $W_3$, $W_2[i,j]$, $W_4[i,j]$ hold for every pair $\{i,j\}$;
\item[(ii)]
if $W\in (\bbb^\circ\meet\ccc^\circ)\join(\bbb^\circ\meet\aaa^\circ)$, then inequalities
$W_1[1,3]$, $W_1[2,4]$, $W_3$, $W_2[i,j]$, $W_4[i,j]$ hold for every pair $\{i,j\}$;
\item[(iii)]
if $W\in (\ccc^\circ\meet\aaa^\circ)\join(\ccc^\circ\meet\beta^\circ)$, then inequalities
$W_1[1,2]$, $W_1[3,4]$, $W_3$, $W_2[i,j]$, $W_4[i,j]$ hold for every pair $\{i,j\}$.
\end{enumerate}
If $W$ is $\xx$-shaped, then the converses also hold.
\end{theorem}

\begin{proof}
We will prove (i) and its converse for $W=\xx(s,t,u)$. The necessity follows from the inclusion
$$
(\aaa^\circ\meet\bbb^\circ)\join(\aaa^\circ\meet\ccc^\circ)
\le
\aaa^\circ\meet (\aaa^\circ\join(\bbb^\circ\meet\ccc^\circ))\meet
(\bbb^\circ\join(\aaa^\circ\meet\ccc^\circ))\meet(\ccc^\circ\join(\aaa^\circ\meet\bbb^\circ)),
$$
by \cite[Proposition 3.3]{han_kye_pe} and Theorem \ref{join-statestep4}.

We prove the converse when $W=\xx(s,t,u)$. To do this, we suppose the following:
\begin{itemize}
\item
$W=\xx(s,t,u)$ satisfies
$W_1[1,4]$, $W_1[2,3]$, $W_3$, $W_2[i,j]$, $W_4[i,j]$ hold for every pair $\{i,j\}$;
\item
$\varrho=\xx(a,b,z) \in (\aaa \join \bbb) \meet (\aaa \join \ccc)$, that is, satisfies $S_2[1,2]$, $S_2[1,3]$,
\end{itemize}
and prove the inequality
\begin{equation}\label{ineqxxxx}
\lan W,\varrho\ran\ge 0.
\end{equation}
If $\varrho\in\aaa$ then we have (\ref{ineqxxxx}) since $W\in\aaa^\circ$ by $W_1[1,4]$ and $W_1[2,3]$.
If $\varrho\in\bbb$ then we have $\varrho\in\bbb\meet(\aaa\join\ccc)$, and so
the inequality (\ref{ineqxxxx}) follows since $W\in\bbb^\circ\join(\ccc^\circ\meet\aaa^\circ)$.
We also have (\ref{ineqxxxx}) when $\varrho\in\ccc$ by the same reasoning.

Therefore, we may assume that $\varrho\notin\aaa$, $\varrho\notin\bbb$ and $\varrho\notin\ccc$.
By $\varrho \notin \alpha$, we may assume
$$
(A 4)~ |z_4|>\sqrt{a_1b_1}.
$$
By the assumption $\varrho\notin\beta$, there are four possibilities
$$
(B1)~ |z_1|>\sqrt{a_3b_3}, \quad (B2)~ |z_2|>\sqrt{a_4b_4},
\quad  (B3)~ |z_3|>\sqrt{a_1b_1}, \quad  (B4)~ |z_4|>\sqrt{a_2b_2}.
$$
We also have the following four cases
$$
(C1)~ |z_1|>\sqrt{a_2b_2}, \quad  (C2)~ |z_2|>\sqrt{a_1b_1},
 \quad (C3)~ |z_3|>\sqrt{a_4b_4},  \quad  (C4)~ |z_4|>\sqrt{a_3b_3},
$$
by the assumption $\varrho\notin\ccc$.
Under the assumption (A4), we have the implications
$(B1) \Rightarrow (C4)$, $(B2) \Rightarrow (C2)$, $(C1) \Rightarrow (B4)$ and $(C3) \Rightarrow (B3)$.
Conversely, the condition (A4) can be implied as
$(B3), (C4) \Rightarrow (A4)$ and $(B4), (C2) \Rightarrow (A4)$.
Hence, it suffices to consider the following eight cases:
$$
\begin{aligned}
&(A4), (B1); \quad &(A4), (B2); \quad &(A4), (C1); \quad &(A4), (B3), (C2); \\
&(A4), (C3); \quad &(B3), (C4); \quad &(B4), (C2); \quad &(A4), (B4), (C4). \\
\end{aligned}
$$
Switching the second and third subsystems interchanges
$$
(A4), (B1) \leftrightarrow (A4), (C1), \quad (A4), (B2) \leftrightarrow (A4), (C3), \quad (B3), (C4) \leftrightarrow (C2), (B4).
$$
Therefore, it suffices to consider the following five cases:
$$
(B3), (C4) ; \quad (A4), (B3), (C2); \quad (A4), (B4), (C4) ; \quad (A4), (B1); \quad (A4), (B2).
$$

[Case I: (B3) and (C4)]. We have
$|z_4| \ge \sqrt{a_3b_3} \ge |z_3| \ge \sqrt{a_1b_1}$, and
\begin{align*}
{\tfrac 12}\lan W, \varrho \ran
& \ge \textstyle\sum_{i=1}^4 \sqrt{s_i t_i} \sqrt{a_i b_i} - |u_i| |z_i| \\
& \ge  \sqrt{s_1t_1}\sqrt{a_1b_1} + \sqrt{s_2t_2}(|z_3|+|z_4|-\sqrt{a_1b_1}) + \sqrt{s_3t_3}\sqrt{a_3b_3} + \sqrt{s_4t_4}|z_4| \\
& \qquad - |u_1|\sqrt{a_1b_1} - (\textstyle\sum_{i=1}^4 \sqrt{s_it_i}-\textstyle\sum_{i \ne 2} |u_i| )
                     (\sqrt{a_1b_1}+\sqrt{a_3b_3}-|z_4|)\\
& \qquad - |u_3||z_3|- |u_4||z_4|
\end{align*}
by $S_2[1,2]$, $W_3$ and $S_2[1,3]$. This is equal to
\begin{align*}
& = (\sqrt{s_1t_1} + 2\sqrt{s_2t_2} + \sqrt{s_3t_3} +2\sqrt{s_4t_4} -|u_1|-|u_3|-2|u_4|) |z_4| \\
& \qquad +(-\textstyle\sum_{i \ne 3} \sqrt{s_it_i} + \sum_{i \ne 2} |u_i|) \sqrt{a_3b_3} +(\sqrt{s_2t_2}- |u_3| ) |z_3| \\
& \qquad +(-2\sqrt{s_2t_2} - \sqrt{s_3t_3} - \sqrt{s_4t_4} +|u_3| +|u_4|) \sqrt{a_1b_1}.
\end{align*}
Therefore, applying $W_{\rm 4b}[1,3]$, we have
\begin{align*}
{\tfrac 12}\lan W, \varrho \ran
& \ge (\sqrt{s_1t_1} + 2\sqrt{s_2t_2} + \sqrt{s_3t_3} +2\sqrt{s_4t_4} -|u_1|-|u_3|-2|u_4|) \sqrt{a_3b_3} \\
& \qquad +(-\textstyle\sum_{i \ne 3} \sqrt{s_it_i} + \sum_{i \ne 2} |u_i| ) \sqrt{a_3b_3}  +(\sqrt{s_2t_2}- |u_3| ) |z_3| \\
& \qquad +(-2\sqrt{s_2t_2} - \sqrt{s_3t_3} - \sqrt{s_4t_4} +|u_3| +|u_4|) \sqrt{a_1b_1}\\
& = (\sqrt{s_2t_2} + \sqrt{s_3t_3} + \sqrt{s_4t_4} -|u_4|) \sqrt{a_3b_3}
    +(\sqrt{s_2t_2}- |u_3| ) |z_3| \\
& \qquad +(-2\sqrt{s_2t_2} - \sqrt{s_3t_3} - \sqrt{s_4t_4} +|u_3| +|u_4|) \sqrt{a_1b_1}.
\end{align*}
By $W_{2}[4,1]$, $W_{\rm 4a}[3,4]$ and $|z_3|\ge \sqrt{a_1b_1}$, this is greater than or equal to
\begin{align*}
& \ge (\sqrt{s_2t_2} + \sqrt{s_3t_3} + \sqrt{s_4t_4} -|u_4|) |z_3| +(\sqrt{s_2t_2}- |u_3| ) |z_3| \\
& \qquad\qquad +(-2\sqrt{s_2t_2} - \sqrt{s_3t_3} - \sqrt{s_4t_4} +|u_3| +|u_4|) |z_3|
= 0.
\end{align*}

[Case II: (A4), (B3) and (C2)]. We have $|z_4|, |z_3|, |z_2| \ge \sqrt{a_1b_1}$. In this case,
we may assume that $|z_2| \ge |z_3|$ by switching the second and the third subsystems.
Put
$$
\lambda_2:=|z_2|-\sqrt{a_1b_1}, \quad \lambda_3:=|z_3|-\sqrt{a_1b_1}, \quad \lambda_4:=|z_4|-\sqrt{a_1b_1},
$$
which are nonnegative by the assumption.
We proceed by considering two subcases.

[Subcase II-1: $\lambda_2 \le \lambda_3 + \lambda_4$]. We have
\begin{align*}
\tfrac 12\lan W, \varrho \ran
& \ge \textstyle\sum_{i=1}^4 \sqrt{s_i t_i} \sqrt{a_i b_i} - |u_i| |z_i| \\
& \ge \sqrt{s_1t_1}\sqrt{a_1b_1}  + \sqrt{s_2t_2}(|z_3|+|z_4|-\sqrt{a_1b_1})\\
& \qquad + \sqrt{s_3t_3}(|z_2|+|z_4|-\sqrt{a_1b_1}) + \sqrt{s_4t_4}|z_4|  -|u_1|\sqrt{a_1b_1} - \textstyle\sum_{i=2}^4 |u_i||z_i|,
\end{align*}
by $S_2[1,2]$ and $S_2[1,3]$. By a direct calculation, this is equal to
\begin{align*}
& =\left(\textstyle\sum_{i=1}^4 \sqrt{s_it_i} - \textstyle\sum_{i=1}^4 |u_i|\right) \sqrt{a_1b_1} \\
& \qquad + (\sqrt{s_3t_3}-|u_2|)\lambda_2
+ (\sqrt{s_2t_2}-|u_3|)\lambda_3
+ (\sqrt{s_2t_2}+\sqrt{s_3t_3}+\sqrt{s_4t_4}-|u_4|)\lambda_4,
\end{align*}
which is greater than or equal to the following
\begin{equation}\label{mmmm}
(\sqrt{s_3t_3}-|u_2|)\lambda_2
+ (\sqrt{s_2t_2}-|u_3|)\lambda_3
+ (\sqrt{s_2t_2}+\sqrt{s_3t_3}+\sqrt{s_4t_4}-|u_4|)\lambda_4,
\end{equation}
by $W_3$.
We note that the sum of the following two terms
$$
\sqrt{s_3t_3}-|u_2|, \qquad \sqrt{s_2t_2}-|u_3|
$$
are nonnegative by $W_1[2,3]$, and so at most one of them is negative possibly. If both of them are nonnegative, then
the proof is complete by $W_2[4,1]$.
If $\sqrt{s_2t_2}-|u_3|<0$, then
we replace $\lambda_3$ in (\ref{mmmm}) by $\lambda_2$ which satisfies $\lambda_2\ge\lambda_3$ by $|z_2| \ge |z_3|$, to get
$$
\begin{aligned}
\tfrac 12\lan W, \varrho \ran
\ge & (\sqrt{s_3t_3}-|u_2|)\lambda_2
+ (\sqrt{s_2t_2}-|u_3|)\lambda_2
+ (\sqrt{s_2t_2}+\sqrt{s_3t_3}+\sqrt{s_4t_4}-|u_4|)\lambda_4 \\
= & (\sqrt{s_2t_2}+\sqrt{s_3t_3}-|u_2|-|u_3|)\lambda_2
+ (\sqrt{s_2t_2}+\sqrt{s_3t_3}+\sqrt{s_4t_4}-|u_4|)\lambda_4,
\end{aligned}
$$
which is nonnegative by $W_1[2,3]$ and $W_2[4,1]$ again.
If $\sqrt{s_3t_3}-|u_2|<0$, then we replace $\lambda_2$ in (\ref{mmmm})
by $\lambda_3+\lambda_4$ to get
$$
\begin{aligned}
\tfrac 12\lan W, \varrho \ran
\ge &(\sqrt{s_3t_3}-|u_2|)(\lambda_3+\lambda_4)
+ (\sqrt{s_2t_2}-|u_3|)\lambda_3
+ (\sqrt{s_2t_2}+\sqrt{s_3t_3}+\sqrt{s_4t_4}-|u_4|)\lambda_4 \\
= & (\sqrt{s_2t_2}+\sqrt{s_3t_3}-|u_2|-|u_3|)\lambda_3
+ (\sqrt{s_2t_2}+2\sqrt{s_3t_3}+\sqrt{s_4t_4}-|u_2|-|u_4|)\lambda_4,
\end{aligned}
$$
which is also nonnegative by $W_1[2,3]$ and $W_{\rm 4a}[2,4]$.

[Subcase II-2: $\lambda_2 \ge \lambda_3 + \lambda_4$]: In this case, we have
\begin{align*}
\tfrac 12\lan W, \varrho \ran
& \ge \textstyle\sum_{i=1}^4 \sqrt{s_i t_i} \sqrt{a_i b_i} - |u_i| |z_i| \\
& \ge \sqrt{s_1t_1}\sqrt{a_1b_1}
 + \sqrt{s_2t_2}|z_2| + \sqrt{s_3t_3}(|z_2|+|z_4|-\sqrt{a_1b_1}) + \sqrt{s_4t_4}|z_4| \\
& \qquad -|u_1|\sqrt{a_1b_1} - \textstyle\sum_{i=2}^4 |u_i||z_i|,
\end{align*}
by $S_2[1,3]$. This is equal to the following:
\begin{align*}
& = \left(\textstyle\sum_{i=1}^4 \sqrt{s_it_i} - \textstyle\sum_{i=1}^4 |u_i|\right) \sqrt{a_1b_1} \\
& \qquad + \left((\sqrt{s_2t_2}+\sqrt{s_3t_3}-|u_2|-|u_3|)+|u_3|\right)\lambda_2
-|u_3|\lambda_3
+ (\sqrt{s_3t_3}+\sqrt{s_4t_4}-|u_4|)\lambda_4,
\end{align*}
which is, by $W_3$ and $W_1[2,3]$, greater than or equal to
\begin{align*}
\ge &\left((\sqrt{s_2t_2}+\sqrt{s_3t_3}-|u_2|-|u_3|)+|u_3|\right)(\lambda_3+\lambda_4)
-|u_3|\lambda_3
+ (\sqrt{s_3t_3}+\sqrt{s_4t_4}-|u_4|)\lambda_4 \\
= &\left(\sqrt{s_2t_2}+\sqrt{s_3t_3}-|u_2|-|u_3|\right)\lambda_3
+ \left(\sqrt{s_2t_2}+2\sqrt{s_3t_3}+\sqrt{s_4t_4}-|u_2|-|u_4|\right)\lambda_4.
\end{align*}
This is nonnegative by $W_1[2,3]$ and $W_{\rm 4a}[2,4]$.

[Case III: (A4), (B4), (C4)]. We have $|z_4| \ge \sqrt{a_1b_1}, \sqrt{a_2b_2}, \sqrt{a_3b_3}$.
We may assume that $\sqrt{a_2b_2} \ge \sqrt{a_3b_3}$ by switching the second and the third subsystems.
Put
$$
\mu_1:=|z_4|-\sqrt{a_1b_1}, \quad \mu_2:=|z_4|-\sqrt{a_2b_2}, \quad \mu_3:=|z_4|-\sqrt{a_3b_3}.
$$
Note that $\mu_2\le\mu_3$ by assumption. We proceed by considering two subcases.

[Subcase III-1: $\mu_3 \le \mu_1 + \mu_2$]. In this case, we have
\begin{align*}
\tfrac 12\lan W, \varrho \ran
& \ge \textstyle\sum_{i=1}^4 \sqrt{s_i t_i} \sqrt{a_i b_i} - |u_i| |z_i| \\
& \ge \textstyle\sum_{i=1}^3 \sqrt{s_it_i}\sqrt{a_ib_i} + \sqrt{s_4t_4}|z_4|  - |u_1|\sqrt{a_1b_1} - |u_2|(\sqrt{a_1b_1}+\sqrt{a_3b_3}-|z_4|)\\
&\qquad - |u_3|(\sqrt{a_1b_1}+\sqrt{a_2b_2}-|z_4|) - |u_4||z_4|,
\end{align*}
by $S_2[1,3]$ and $S_2[1,2]$. By a direct computation, this becomes
\begin{equation}\label{estimate}
\begin{aligned}
& = \left(\textstyle\sum_{i=1}^4 \sqrt{s_it_i} - \textstyle\sum_{i=1}^4 |u_i|\right) \sqrt{a_1b_1} \\
& \qquad + (\textstyle\sum_{i \ne 1}\sqrt{s_it_i}-|u_4|)\mu_1 + (-\sqrt{s_2t_2}+|u_3|)\mu_2 + (-\sqrt{s_3t_3}+|u_2|)\mu_3.
\end{aligned}
\end{equation}
If $-\sqrt{s_2t_2}+|u_3| < 0$, then this is, by $W_2[4,1]$ and $\mu_2 \le \mu_3$, greater
than or equal to the following
$$
\begin{aligned}
&\ge
\left(\textstyle\sum_{i=1}^4 \sqrt{s_it_i} - \textstyle\sum_{i=1}^4 |u_i|\right) \sqrt{a_1b_1}
     + (-\sqrt{s_2t_2}+|u_3|-\sqrt{s_3t_3}+|u_2|)\mu_3\\
&\ge
\left(\textstyle\sum_{i=1}^4 \sqrt{s_it_i} - \textstyle\sum_{i=1}^4 |u_i|\right) \sqrt{a_1b_1}
     + (-\sqrt{s_2t_2}+|u_3|-\sqrt{s_3t_3}+|u_2|)\sqrt{a_1b_1}\\
&=(\sqrt{s_1t_1}+\sqrt{s_4t_4}-|u_1|-|u_4|)\sqrt{a_1b_1}
\end{aligned}
$$
by $W_1[2,3]$ and $\sqrt{a_1b_1}\ge \mu_3$ using $S_2[1,3]$. This is nonnegative by
$W_1[1,4]$.
In the case of $-\sqrt{s_2t_2}+|u_3| \ge 0$,
the term (\ref{estimate}) is equal to
\begin{align*}
&= \left(\textstyle\sum_{i=1}^4 \sqrt{s_it_i} - \textstyle\sum_{i=1}^4 |u_i|\right) \sqrt{a_1b_1}
        +  (\textstyle\sum_{i \ne 1}\sqrt{s_it_i}-|u_4|)\mu_1 \\
& \qquad +  (-\sqrt{s_2t_2}+|u_3|)\mu_2 + (-\sqrt{s_2t_2}-\sqrt{s_3t_3}+|u_2|+|u_3|)\mu_3 + (\sqrt{s_2t_2}-|u_3|)\mu_3,
\end{align*}
which is greater than or equal to the following
\begin{align*}
& \ge \left(\textstyle\sum_{i=1}^4 \sqrt{s_it_i} - \textstyle\sum_{i=1}^4 |u_i|\right) \sqrt{a_1b_1}
         +  (\textstyle\sum_{i \ne 1}\sqrt{s_it_i}-|u_4|)\mu_1 \\
& \qquad +  (-\sqrt{s_2t_2}+|u_3|)\mu_2 + (-\sqrt{s_2t_2}-\sqrt{s_3t_3}+|u_2|+|u_3|)\sqrt{a_1b_1} + (\sqrt{s_2t_2}-|u_3|)(\mu_1+\mu_2)
\end{align*}
by $W_1[2,3]$ and $\sqrt{a_1b_1}\ge \mu_3$. This becomes
\begin{align*}
& = (\sqrt{s_1t_1}+\sqrt{s_4t_4}-|u_1|-|u_4|)\sqrt{a_1b_1} + (2\sqrt{s_2t_2}+\sqrt{s_3t_3}+\sqrt{s_4t_4}-|u_3|-|u_4|)\mu_1,
\end{align*}
which is nonnegative by $W_1[1,4]$ and $W_{\rm 4a}[3,4]$.

[Subcase III-2: $\mu_3 \ge \mu_1 + \mu_2$]. We use $S_2[1,3]$ to get the inequality
\begin{align*}
\tfrac 12\lan W, \varrho \ran
& \ge \textstyle\sum_{i=1}^4 \sqrt{s_i t_i} \sqrt{a_i b_i} - |u_i| |z_i| \\
& \ge \textstyle\sum_{i=1}^3 \sqrt{s_it_i}\sqrt{a_ib_i} + \sqrt{s_4t_4}|z_4| \\
& \qquad - |u_1|\sqrt{a_1b_1} - |u_2|(\sqrt{a_1b_1}+\sqrt{a_3b_3}-|z_4|)
- |u_3|\sqrt{a_3b_3} - |u_4||z_4|,
\end{align*}
which is equal to
\begin{align*}
& = \left(\textstyle\sum_{i=1}^4 \sqrt{s_it_i} - \textstyle\sum_{i=1}^4 |u_i|\right) \sqrt{a_1b_1}
               + (\sqrt{s_2t_2}+\sqrt{s_3t_3}+\sqrt{s_4t_4}-|u_3|-|u_4|)\mu_1 \\
& \qquad  - \sqrt{s_2t_2}\mu_2 + \sqrt{s_2t_2}\mu_3 + (-\sqrt{s_2t_2}-\sqrt{s_3t_3}+|u_2|+|u_3|)\mu_3,
\end{align*}
by a direct calculation. Using $W_1[2,3]$ and $\mu_3\le
\sqrt{a_1b_1}$, we continue
\begin{align*}
& \ge \left(\textstyle\sum_{i=1}^4 \sqrt{s_it_i} - \textstyle\sum_{i=1}^4 |u_i|\right) \sqrt{a_1b_1}
             + (\sqrt{s_2t_2}+\sqrt{s_3t_3}+\sqrt{s_4t_4}-|u_3|-|u_4|)\mu_1 \\
& \qquad  - \sqrt{s_2t_2}\mu_2 + \sqrt{s_2t_2}(\mu_1+\mu_2) + (-\sqrt{s_2t_2}-\sqrt{s_3t_3}+|u_2|+|u_3|)\sqrt{a_1b_1} \\
& = (\sqrt{s_1t_1}+\sqrt{s_4t_4}-|u_1|-|u_4|)\sqrt{a_1b_1} + (2\sqrt{s_2t_2}+\sqrt{s_3t_3}+\sqrt{s_4t_4}-|u_3|-|u_4|)\mu_1,
\end{align*}
which is nonnegative by $W_1[1,4]$ and $W_{\rm 4a}[3,4]$.

[Case IV: (A4), (B1)]. We have $|z_4| \ge \sqrt{a_1b_1} \ge |z_1|
\ge \sqrt{a_3b_3}$. By Case III, we may suppose that
$\sqrt{a_2b_2} \ge |z_4|$. By $W_3$ and $S_2[1,3]$, we have
\begin{align*}
\tfrac 12\lan W, \varrho \ran
& \ge \textstyle\sum_{i=1}^4 \sqrt{s_i t_i} \sqrt{a_i b_i} - |u_i| |z_i| \\
& \ge  \sqrt{s_1t_1}\sqrt{a_1b_1} + \sqrt{s_2t_2}|z_4| + \sqrt{s_3t_3}\sqrt{a_3b_3} + \sqrt{s_4t_4}|z_4| - |u_1||z_1|\\
& \quad  - (\textstyle\sum_{i=1}^4 \sqrt{s_it_i}-\textstyle\sum_{i \ne 2} |u_i| )
                (\sqrt{a_1b_1}+\sqrt{a_3b_3}-|z_4|)- |u_3|\sqrt{a_3b_3}- |u_4||z_4|
\end{align*}
which becomes
\begin{align*}
& = (\sqrt{s_1t_1} + 2\sqrt{s_2t_2} + \sqrt{s_3t_3} +2\sqrt{s_4t_4} -|u_1|-|u_3|-2|u_4|) |z_4|\\
& \qquad +(-\textstyle\sum_{i \ne 1} \sqrt{s_it_i} + \textstyle\sum_{i \ne 2} |u_i| ) \sqrt{a_1b_1} \\
& \qquad - |u_1| |z_1|
 +\left((-\sqrt{s_1t_1}-\sqrt{s_4t_4} + |u_1|+|u_4|)-\sqrt{s_2t_2} \right) \sqrt{a_3b_3}.
\end{align*}
Using $W_{\rm 4b}[1,3]$ and $W_1[1,4]$, we continue
\begin{align*}
& \ge (\sqrt{s_1t_1} + 2\sqrt{s_2t_2} + \sqrt{s_3t_3} +2\sqrt{s_4t_4} -|u_1|-|u_3|-2|u_4|) \sqrt{a_1b_1} \\
& \qquad +(-\textstyle\sum_{i \ne 1} \sqrt{s_it_i} + \textstyle\sum_{i \ne 2} |u_i| ) \sqrt{a_1b_1} \\
& \qquad - |u_1| |z_1|  +\left((-\sqrt{s_1t_1}-\sqrt{s_4t_4} + |u_1|+|u_4|)-\sqrt{s_2t_2} \right) |z_1|,
\end{align*}
which is equal to
\begin{align*}
 =(\sqrt{s_1t_1}  + \sqrt{s_2t_2}+\sqrt{s_4t_4} -|u_4|) (\sqrt{a_1b_1}-|z_1|).
\end{align*}
This is nonnegative by $W_2[4,3]$.

[Case V: (A4), (B2)]. In this case, we have $|z_2| \ge \sqrt{a_4b_4} \ge |z_4| \ge \sqrt{a_1b_1}$].
By Case II, we may suppose that $\sqrt{a_1b_1} \ge |z_3|$. Using $S_2[1,3]$, we proceed
\begin{align*}
\tfrac 12\lan W, \varrho \ran
& \ge \textstyle\sum_{i=1}^4 \sqrt{s_i t_i} \sqrt{a_i b_i} - |u_i| |z_i| \\
& \ge \sqrt{s_1t_1}\sqrt{a_1b_1} + \sqrt{s_2t_2}|z_2| + \sqrt{s_3t_3}(|z_2|+|z_4|-\sqrt{a_1b_1}) + \sqrt{s_4t_4}\sqrt{a_4b_4} \\
& \qquad - |u_1|\sqrt{a_1b_1} - |u_2| |z_2| - |u_3|\sqrt{a_1b_1}- |u_4||z_4| \\
& = \left((\sqrt{s_2t_2} + \sqrt{s_3t_3} -|u_2|-|u_3|)+|u_3|\right) |z_2| \\
& \qquad +\sqrt{s_4t_4} \sqrt{a_4b_4}  +(\sqrt{s_3t_3}-|u_4|) |z_4| +(\sqrt{s_1t_1}-\sqrt{s_3t_3} - |u_1|-|u_3|) \sqrt{a_1b_1}.
\end{align*}
By $W_1[2,3]$ and $|z_2|\ge |z_4|$, we have
\begin{align*}
\tfrac 12\lan W, \varrho \ran
& \ge \left((\sqrt{s_2t_2} + \sqrt{s_3t_3} -|u_2|-|u_3|)+|u_3|\right) |z_4| \\
& \qquad +\sqrt{s_4t_4} |z_4|  +(\sqrt{s_3t_3}-|u_4|) |z_4|  +(\sqrt{s_1t_1}-\sqrt{s_3t_3} - |u_1|-|u_3|) \sqrt{a_1b_1}
\end{align*}
This is equal to
$$
=(\sqrt{s_2t_2} + 2\sqrt{s_3t_3} + \sqrt{s_4t_4} -|u_2|-|u_4|) |z_4|
+(\sqrt{s_1t_1}-\sqrt{s_3t_3} - |u_1|-|u_3|) \sqrt{a_1b_1}.
$$
Using $W_{\rm 4a}[2,4]$, we may replace $|z_4|$ by $\sqrt{a_1b_1}$, to get
\begin{align*}
\tfrac 12\lan W, \varrho \ran
& \ge (\sqrt{s_2t_2} + 2\sqrt{s_3t_3} + \sqrt{s_4t_4} -|u_2|-|u_4|) \sqrt{a_1b_1} \\
& \qquad +(\sqrt{s_1t_1}-\sqrt{s_3t_3} - |u_1|-|u_3|) \sqrt{a_1b_1},
\end{align*}
which is equal to $\left(\sum_{i=1}^4\sqrt{s_it_i}-|u_i|\right) \sqrt{a_1b_1}\ge 0$ by $W_3$.
This completes the proof.
\end{proof}

We recall the definition of $Q$ in Section 3:
$$
Q:=[\aaa\meet(\bbb\join\ccc)]\join[\bbb\meet(\ccc\join\aaa)]\join[\ccc\meet(\aaa\join\bbb)].
$$
Then we have
$$
Q^\circ=[\aaa^\circ\join(\bbb^\circ\meet\ccc^\circ)]\meet[\bbb^\circ\join(\ccc^\circ\meet\aaa^\circ)]
   \meet[\ccc^\circ\join(\aaa^\circ\meet\bbb^\circ)].
$$
For a given $\sigma\in{\mathcal L}$, we define
$$
\begin{aligned}
\sigma^\circ_\xx := \sigma^\circ\cap \xx
&=\{W\in\xx : \lan W,\varrho\ran\ge 0\ {\text{\rm for every}}\ \varrho\in\sigma\}\\
&=\{W\in\xx : \lan W,\varrho\ran\ge 0\ {\text{\rm for every}}\ \varrho\in\sigma_\xx\}.
\end{aligned}
$$
The last identity follows from (\ref{x-part-w}), and we see that the {\xx}-part $(\sigma^\circ)_\xx$ of $\sigma^\circ$
coincides with the dual of $(\sigma_\xx)^\circ$ of $\sigma_\xx$ in the space $\xx$.
By Theorem \ref{dist11122}, we have the
following:

\begin{corollary}
We have the following identities
$$
\begin{aligned}
(\aaa^\circ_\xx\meet\bbb^\circ_\xx)\join(\aaa^\circ_\xx\meet\ccc^\circ_\xx)&=\aaa^\circ_\xx\meet Q^\circ_\xx,\\
(\bbb^\circ_\xx\meet\ccc^\circ_\xx)\join(\bbb^\circ_\xx\meet\aaa^\circ_\xx)&=\bbb^\circ_\xx\meet Q^\circ_\xx,\\
(\ccc^\circ_\xx\meet\aaa^\circ_\xx)\join(\ccc^\circ_\xx\meet\bbb^\circ_\xx)&=\ccc^\circ_\xx\meet Q^\circ_\xx.
\end{aligned}
$$
\end{corollary}

Taking the dual cones in the vector space $\xx$, we get the following identities among the
generators $\aaa_\xx$, $\bbb_\xx$ and $\ccc_\xx$ of the lattice
${\mathcal L}_\xx$, which is the lattice theoretic dual identities of those in Corollary \ref{identity_P}.

\begin{corollary}\label{identity_Q}
We have the following identities
$$
\begin{aligned}
(\aaa_\xx\join\bbb_\xx)\meet(\aaa_\xx\join\ccc_\xx)&=\aaa_\xx\join Q_\xx,\\
(\bbb_\xx\join\ccc_\xx)\meet(\bbb_\xx\join\aaa_\xx)&=\bbb_\xx\join Q_\xx,\\
(\ccc_\xx\join\aaa_\xx)\meet(\ccc_\xx\join\bbb_\xx)&=\ccc_\xx\join Q_\xx.
\end{aligned}
$$
\end{corollary}

\section{Greenberger-Horne-Zeilinger diagonal states}

We recall that an $\xx$-state $\xx(a,b,z)$ is GHZ diagonal if and
only if $a=b$ and $z\in\mathbb R^4$. More generally, we will say
that an \xx-shaped matrix $\xx(s,t,u)$ is GHZ diagonal if and only
if $s=t$ and $u\in\mathbb R^4$. In this section, we exhibit all the
GHZ diagonal states which belong to the convex cones considered in
this paper. For this purpose, we search for extreme rays of the
corresponding convex cones. We consider the following local
operation
$$
U=\begin{pmatrix}0&1\\1&0\end{pmatrix} \otimes \begin{pmatrix}0&1\\1&0\end{pmatrix} \otimes \begin{pmatrix}0&1\\1&0\end{pmatrix},
$$
which interchanges $|0\ran$ and $|1\ran$ in each subsystem $M_2 \otimes M_2 \otimes M_2$.
For an $\xx$-shaped self-adjoint three qubit matrix $W=\xx(s,t,u)$, we define
$$
\widetilde W:=UWU^*=\xx(t,s,\bar u).
$$
We also define
$$
W_\GHZ:={W + \widetilde W \over 2} = X\left({s+t \over 2}, {s+t \over 2}, {\rm Re} ~u\right).
$$
Then an \xx-shaped matrix $W$ is GHZ diagonal if and only if $W=W_\GHZ$, and the identity
\begin{equation}\label{kjhftyd}
\lan W_\GHZ,\varrho\ran=\lan W,\varrho_\GHZ\ran
\end{equation}
holds for every $\xx$-shaped $W$ and $\varrho$, as in the proof of
\cite[Theorem 3.2]{han_kye_GHZ}.
We also define $W_\GHZ=(W_\xx)_\GHZ$ for a general three qubit self-adjoint matrix $W$.
Then, $W$ is GHZ diagonal if and only if $W=W_\GHZ$. Further, the identity (\ref{kjhftyd}) also holds
for every self-adjoint $\sigma$ and $W$.
We denote by $V$ the eight dimensional real vector space consisting of GHZ diagonal matrices.
If we take a convex cone $\sigma\in{\mathcal L}$, then
the dual cone $(\sigma\cap V)^\circ$ of $\sigma\cap V$ in the vector space $V$ coincides
with $\sigma^\circ\cap V$, because
$$
\lan W, \varrho \ran = \lan W_\GHZ, \varrho \ran = \lan W, \varrho_\GHZ \ran \ge 0,
$$
for $W \in (\sigma\cap V)^\circ$ and $\varrho \in \sigma$.
For GHZ diagonal states, we will use the notation
\begin{equation}\label{GHZ_sym}
\xx(a,a,z)=\xx{a_1~a_2~a_3~a_4\choose z_1~z_2~z_3~z_4},
\end{equation}
with real variables $a_i$ and $z_i$ for $i=1,2,3,4$.

We also note that all the conditions $S_1$, $S_2$, $S_3$ and $S_4$ for GHZ diagonal states are determined by
finitely many linear inequalities with respect to eight real variables. For example,
the inequality $S_4[i,j|k,\ell]$ is actually the combination of the following eight linear inequalities:
$$
a_i+a_j\ge \pm z_k \pm z_\ell,\quad
a_k+a_\ell\ge \pm z_i \pm z_j
$$
for a GHZ diagonal state $\xx{a_1~a_2~a_3~a_4\choose z_1~z_2~z_3~z_4}$.
For a given convex cone $\sigma\in{\mathcal L}$ considered in this paper,
the convex cone $\sigma\cap V$ is determined by finitely many hyperplanes which cover all
maximal faces of $\sigma\cap V$. The same is true for $\sigma^\circ\cap V$.
If $\varrho$ is an extreme ray of $\sigma\cap V$,
then its dual face
$$
\varrho^\prime:=\{W\in (\sigma\cap V)^\circ: \lan W,\varrho\ran=0\}
$$
of the dual cone $\sigma^\circ\cap V$ is a maximal face. See
\cite[Theorem 5.3]{han_kye_pe}. Therefore, we conclude that
$\sigma\cap V$ has only finitely many extreme rays, which must be
orthogonal to a hyperplane determining a maximal face
of the dual cone $\sigma^\circ\cap V$. Considering the
coefficients of the characteristic linear inequalities of
$\sigma^\circ\cap V$, it is straightforward to find all the
candidates  of extreme rays of the convex cone $\sigma\cap V$. We
denote by $\CE_\sigma$ the set of candidates for extreme rays of
$\sigma\cap V$ arising in this way. We take a chain of convex cones
from the diagram (\ref{diag_con}), and list up in {\sc Table 2} all
the candidates in $\CE_\sigma$. Note that the diagonal states
$$
\Delta=\{\textstyle
\xx{1\,0\,0\,0\choose 0\,0\,0\,0},\
\xx{0\,1\,0\,0\choose 0\,0\,0\,0},\
\xx{0\,0\,1\,0\choose 0\,0\,0\,0},\
\xx{0\,0\,0\,1\choose 0\,0\,0\,0}
\}
$$
are extreme rays of all the convex cones.

At this moment, it must  be pointed out that several
inequalities of Theorem \ref{dist11122} are redundant. Let us
consider the criteria of
$(\aaa^\circ_\xx\meet\bbb^\circ_\xx)\join(\aaa^\circ_\xx\meet\ccc^\circ_\xx)$
in Theorem \ref{dist11122}. One may easily see that $W_1[1,4]$ and
$W_1[2,3]$ imply the inequalities $W_3$, $W_2[1,2]$, $W_2[1,3]$,
$W_2[2,4]$, $W_2[3,4]$ and $W_4[1,4]$, $W_4[2,3]$. We may also add one of
$W_1[1,4]$, $W_1[2,3]$ and one of $W_{4{\rm a}}[1,2]$, $W_{4{\rm
a}}[1,3]$, $W_{4{\rm a}}[2,4]$, $W_{4{\rm a}}[3,4]$, to get the
inequalities $W_{4{\rm b}}[1,2]$, $W_{4{\rm b}}[1,3]$, $W_{4{\rm
b}}[2,4]$ and $W_{4{\rm b}}[3,4]$.
For an example, we get a part of $W_{4{\rm b}}[1,2]$ from
$W_1[2,3]$ and $W_{4{\rm a}}[1,3]$.
In {\sc Table 2}, the states determined by these redundant inequalities are excluded.

\renewcommand{\arraystretch}{1.5}
\begin{table}
\begin{center}
\begin{tabular}{|c|l|c|}
\hline $\sigma$ &
\qquad\qquad\qquad\qquad\qquad$\CE_\sigma$
   & number of $\CE_\sigma$\\
\hline $\aaa\meet\bbb\meet\ccc$
  &$\Delta$,\ $\xx{~1~~1~~1~~1\choose \mathsmaller{\pm}\!1~\mathsmaller{\pm}\!1~\mathsmaller{\pm}\!1~\mathsmaller{\pm}\!1}$
  &$4+16=20$\\
\hline $\aaa\meet\bbb$
  &add $\xx{~1\,0\,1\,1\choose \mathsmaller{\pm}\!1\,0\,0\,0}$\ \
   $\xx{0\,~1\,1\,1\choose 0\,\mathsmaller{\pm}\!1\,0\,0}$\ \
   $\xx{1\,1\,~1\,0\choose 0\,0\,\mathsmaller{\pm}\!1\,0}$\ \
   $\xx{1\,1\,0\,~1\choose 0\,0\,0\,\mathsmaller{\pm}\!1}$
  &$20+8=28$\\
\hline $(\aaa\meet\bbb)\join(\aaa\meet\ccc)$
  &add $\xx{~1\,1\,0\,1\choose \mathsmaller{\pm}\!1\,0\,0\,0}$\ \
   $\xx{0\,1\,~1\,1\choose 0\,0\,\mathsmaller{\pm}\!1\,0}$\ \
   $\xx{1\,~1\,1\,0\choose 0\,\mathsmaller{\pm}\!1\,0\,0}$\ \
   $\xx{1\,0\,1\,~1\choose 0\,0\,0\,\mathsmaller{\pm}\!1}$
   &$28+8=36$\\
\hline $\aaa\meet(\bbb\join\ccc)$
  &add $\xx{~1\,2\,0\,~1\choose \mathsmaller{\pm}\!1\,0\,0\mathsmaller{\pm}\!1}$\ \ \
   $\xx{~1\,0\,2\,~1\choose \mathsmaller{\pm}\!1\,0\,0\,\mathsmaller{\pm}\!1}$\ \ \
   $\xx{2\,~1\,~1\,0\choose 0\,\mathsmaller{\pm}\!1\,\mathsmaller{\pm}\!1\,0}$\ \ \
   $\xx{0\,~1\,~1\,2\choose 0\,\mathsmaller{\pm}\!1\,\mathsmaller{\pm}\!1\,0}$
  &$36+48=84$\\
& \phantom{add} $\xx{~1\,~2\,2\,~1\choose \mathsmaller{\pm}\!1\,\mathsmaller{\pm}\!2\,0\,\mathsmaller{\pm}\!1}$\ \
   $\xx{~1\,2\,~2\,~1\choose \mathsmaller{\pm}\!1\,0\,\mathsmaller{\pm}\!2\,\mathsmaller{\pm}\!1}$\ \
   $\xx{~2\,~1\,~1\,2\choose \mathsmaller{\pm}\!2\,\mathsmaller{\pm}\!1\,\mathsmaller{\pm}\!1\,0}$\ \
   $\xx{2\,~1\,~1\,~2\choose 0\,\mathsmaller{\pm}\!1\,\mathsmaller{\pm}\!1\,\mathsmaller{\pm}\!2}$
  &\\
\hline\hline $\aaa$
  &$\Delta$,\ $\xx{~1\,0\,0\,~1\choose \mathsmaller{\pm}\!1\,0\,0\,\mathsmaller{\pm}\!1}$\ \
   $\xx{0\,~1\,~1\,0\choose 0\,\mathsmaller{\pm}\!1\,\mathsmaller{\pm}\!1\,0}$
   &$4+8=12$\\
\hline $\aaa\join(\bbb\meet\ccc)$
  &add $\xx{~1\,1\,1\,0\choose \mathsmaller{\pm}\!1\,0\,0\,0}$, \
   $\xx{0\,1\,1\,~1\choose 0\,0\,0\,\mathsmaller{\pm}\!1}$\ \
   $\xx{1\,~1\,0\,1\choose 0\,\mathsmaller{\pm}\!1\,0\,0}$\ \
   $\xx{1\,0\,~1\,1\choose 0\,0\,\mathsmaller{\pm}\!1\,0}$
  &$12+8=20$\\
\hline $(\aaa\join\bbb)\meet(\aaa\join\ccc)$
  &add $\xx{~1\,~1\,2\,0\choose \mathsmaller{\pm}\!1\,\mathsmaller{\pm}\!1\,0\,0}$\ \
   $\xx{~1\,2\,~1\,0\choose \mathsmaller{\pm}\!1\,0\,\mathsmaller{\pm}\!1\,0}$\ \
   $\xx{2\,~1\,0\,~1\choose 0\,\mathsmaller{\pm}\!1\,0\,\mathsmaller{\pm}\!1}$\ \
   $\xx{2\,0\,~1\,~1\choose 0\,0\,\mathsmaller{\pm}\!1\,\mathsmaller{\pm}\!1}$
  &$20+32=52$\\
 &\phantom{add} $\xx{~1\,~1\,0\,2\choose \mathsmaller{\pm}\!1\,\mathsmaller{\pm}\!1\,0\,0}$\ \
   $\xx{~1\,0\,~1\,2\choose \mathsmaller{\pm}\!1\,0\,\mathsmaller{\pm}\!1\,0}$\ \
   $\xx{0\,~1\,2\,~1\choose 0\,\mathsmaller{\pm}\!1\,0\,\mathsmaller{\pm}\!1}$\ \
   $\xx{0\,2\,~1\,~1\choose 0\,0\,\mathsmaller{\pm}\!1\,\mathsmaller{\pm}\!1}$
 &\\
\hline\hline $\aaa\join\bbb$
  &$\Delta$,\ $\xx{~1\,0\,0\,~1\choose \mathsmaller{\pm}\!1\,0\,0\,\mathsmaller{\pm}\!1}$\ \
   $\xx{0\,~1\,~1\,0\choose 0\,\mathsmaller{\pm}\!1\,\mathsmaller{\pm}\!1\,0}$\ \
   $\xx{~1\,0\,~1\,0\choose \mathsmaller{\pm}\!1\,0\,\mathsmaller{\pm}\!1\,0}$\ \
   $\xx{0\,~1\,0\,~1\choose 0\,\mathsmaller{\pm}\!1\,0\,\mathsmaller{\pm}\!1}$
  &$4+16=20$\\
\hline $\aaa\join\bbb\join\ccc$
  &add $\xx{~1\,~1\,0\,0\choose \mathsmaller{\pm}\!1\,\mathsmaller{\pm}\!1\,0\,0}$\ \
   $\xx{0\,0\,~1\,~1\choose 0\,0\,\mathsmaller{\pm}\!1\,\mathsmaller{\pm}\!1}$
  &$20+8=28$\\
\hline
\end{tabular}
\caption{The convex cones in the first column are increasing
downward with respect to inclusion. For each convex cone $\sigma$ in
the first column, we list up all the states in $\CE_\sigma$ which
arise from characteristic inequalities for the dual cone
$\sigma^\circ$. The collection $\CE_\sigma$ contain all extreme rays
of the convex cone $\sigma\cap V$ consisting of GHZ diagonal states.
We prove that all the states in $\CE_\sigma$ are indeed extreme.}
\end{center}
 \end{table}

In order to prove that all the candidates are actually extreme rays,
we use again duality in the vector space $V$. Suppose that $\sigma$
is a convex cone in {\sc Table 2} with the set $\CE_\sigma$
containing all the candidates of extreme rays of $\sigma\cap V$. We
first recall that the bi-dual face
$$
(\varrho^\prime)^\prime=\{\psi\in\sigma\cap V:
\lan W,\psi\ran=0\ {\text{\rm for every}}\ W\in  \varrho^\prime\}
$$
is the smallest exposed face of $\sigma\cap V$ containing $\rho$.
Thus, a state $\varrho\in\sigma\cap V$ is extreme if its bi-dual
face $\varrho^{\prime\prime}$ consists of nonnegative scalar
multiples of $\varrho$. Because $\CE_\sigma$ contains all extreme
rays of $\sigma\cap V$, we see that this is the case if for each
$\psi\in\CE_\sigma\setminus \{\varrho\}$, there exists a witness
$W_\psi\in V$ such that $\lan W_\psi,\varrho\ran=0$ and $\lan
W_\psi,\psi\ran >0$. By taking $W=\sum_\psi W_\psi$, this is
equivalent to the existence of a witness $W\in V$ satisfying
\begin{equation}\label{condi}
\lan W,\varrho\ran=0,\qquad \lan W,\psi\ran >0\ {\text{\rm for
each}}\ \psi\in\CE_\sigma\setminus \{\varrho\},
\end{equation}
which is a seemingly stronger condition. Note that (\ref{condi}) actually tells us that
$W$ belongs to $\sigma^\circ\cap V=(\sigma\cap V)^\circ$, since every state in $\sigma\cap V$ is a convex combination of states in $\CE_\sigma$.
Geometrically, it is clear that the condition (\ref{condi})
is equivalent to the claim that the dual face $\varrho^\prime$ is a maximal face of the dual cone $\sigma^\circ\cap V$.

To prove that all the states in $\CE_\sigma$ are extreme, we begin
with the convex cone $\sigma=\aaa\meet(\bbb\join\ccc)$. For
$\varrho=\xx{~1~0~0~0\choose ~0~0~0~0}$, one can check the witness
{$W=\xx{~0~1~1~1\choose ~0~0~0~0}$} works. For other states like
$$
\textstyle
\xx{~1~ 1~ 1~ 1\choose ~1~ 1~ 1 ~1},\quad
\xx{~1~ 0~ 1~ 1\choose ~1~ 0~ 0~ 0},\quad
\xx{~1~ 1~ 0~ 1\choose ~1~ 0~ 0~ 0},\quad
\xx{~1~ 2~ 0~ 1\choose ~1~ 0~ 0~ 1},\quad
\xx{~1~ 2~ 2~ 1\choose ~1~ 2~ 0~ 1},
$$
the following witnesses
$$
\textstyle
\xx{~1\,~1\,~1\,~1\choose \shortminus1\, \shortminus1\, \shortminus1\, \shortminus1},\quad
\xx{\,1~3~1~1\choose
\shortminus3~0~0~0},\quad
\xx{\,1~1~3~1\choose
\shortminus3~0~0~0},\quad
\xx{\,1~1~2~\,1\choose
\shortminus2~0~0~\shortminus2},\quad
\xx{~2~\,1~\,1~~2\choose
\shortminus1~\shortminus3~0~\shortminus1}
$$
satisfy the condition (\ref{condi}).
When an anti-diagonal of a candidate has minus sign, we take plus sign on the corresponding entry of the witness.
The others are similar by symmetry.

We also recall that if $C_1$ and $C_2$ are convex cones and
$x\in C_1\subset C_2$ is extreme in $C_2$ then it is also extreme in $C_1$.
Therefore, all the candidates in {\sc Table 2} are really extreme rays for the convex cones
$\aaa\meet\bbb\meet\ccc$, $\aaa\meet\bbb$ and $(\aaa\meet\bbb)\join(\aaa\meet\ccc)$.

It was actually shown in \cite{han_kye_pe} that all the candidates in {\sc Table 2} are extreme for
$\aaa\meet\bbb\meet\ccc$, $\aaa\meet\bbb$, $\aaa$, $\aaa\join \bbb$ and $\aaa\join\bbb\join\ccc$.
Therefore, it remains to consider the convex cone $(\aaa\join\bbb)\meet(\aaa\join\ccc)$.
Again, for the states
$$
\textstyle
\xx{1~0~0~1\choose 1~0~0~1},\qquad
\xx{1~1~1~0\choose 1~0~0~0},\qquad
\xx{1~1~2~0\choose 1~1~0~0},
$$
the witnesses
$$
\textstyle
\xx{\,1~1~1~\,1\choose
\shortminus1~0~0~\shortminus1}, \qquad
\xx{~1~1~1~3\choose
\shortminus3~0~0~0}, \qquad
\xx{\,1~~5~1~3\choose
\shortminus3~ \shortminus5 ~0~0}
$$
satisfy (\ref{condi}), respectively. This shows that all the candidates in {\sc Table 2} are really
extreme rays of the corresponding convex cones.

\section{Conclusion}

In this paper, we gave criteria for the convex cones
listed in (\ref{list}). In the plain terminologies, a state
$\varrho$ belongs to the convex cone
$(\aaa\meet\bbb)\join(\aaa\meet\ccc)$ if and only if it is a mixture of a simultaneously $A$-$BC$ and $B$-$CA$ bi-separable state
and a simultaneously $A$-$BC$ and $C$-$BA$ bi-separable state.
We gave a necessary condition for a three qubit state $\varrho$
to have this property in terms of diagonal and anti-diagonal entries of $\varrho$, and showed
that this condition is also sufficient when $\varrho$ is $\xx$-shaped.
We also found all the GHZ diagonal states which distinguish kinds of partial separability.
For example, we can read out from {\sc Table 2} all the extremal GHZ diagonal states which violate distributive rules.


It was asked in \cite{han_kye_szalay} whether the lattice ${\mathcal L}$ is free or not.
The identities in Corollary \ref{identity_P} actually shows that the lattice ${\mathcal L}_\xx$ is not free, because
we have exhibited a lattice generated by three elements which give rise to a strict inequality in (\ref{aaaaa}).
It is natural to ask whether the identity holds in (\ref{aaaaa}) when $(x,y,z)=(\aaa,\bbb,\ccc)$.
This is to equivalent to ask whether the following properties
\begin{itemize}
\item
$\varrho$ is $A$-$BC$ separable;
\item
$\varrho$ is a mixture of $B$-$CA$ separable state and a simultaneously $C$-$AB$ and $A$-$BC$ separable state;
\item
$\varrho$ is a mixture of $C$-$BA$ separable state and a simultaneously $B$-$AC$ and $A$-$BC$ separable state,
\end{itemize}
for a three qubit state $\varrho$ implies that $\varrho$ is
a mixture of a simultaneously $A$-$BC$ and $B$-$CA$ bi-separable state
and a simultaneously $A$-$BC$ and $C$-$BA$ bi-separable state. We have seen in this paper that this is the case when
$\varrho$ is $\xx$-shaped. This question must be related with the question
whether the lattice ${\mathcal L}$ is free or not.



\begin{thebibliography}{99}



\bibitem{bdmsst}
C. H. Bennett, D. P. DiVincenzo, T. Mor, P. W. Shor, J. A. Smolin and B. M. Terhal,
\it Unextendible product bases and bound entanglement,
\rm Phys. Rev. Lett. \bf 82 \rm (1999), 5385--5388.

\bibitem{dur-cirac}
W. D\" ur and J. I. Cirac,
\it Classification of multi-qubit mixed states: separability and distillability properties,
\rm Phys. Rev. A {\bf 61} (2000), 042314.

\bibitem{dur-cirac-tarrach}
W. D\" ur, J. I. Cirac and R. Tarrach,
\it Separability and Distillability of Multiparticle Quantum Systems,
\rm Phys. Rev. Lett. {\bf 83} (1999), 3562--3565.

\bibitem{seevinck-uffink}
M. Seevinck and J. Uffink,
\it Partial separability and etanglement criteria for multiqubit quantum states,
\rm Phys. Rev. A {\bf 78} (2008), 032101.

\bibitem{acin}
A. Acin, D. Bru\ss, M. Lewenstein and A. Sanpera,
\it Classification of mixed three-qubit states,
\rm Phys. Rev. Lett. {\bf 87} (2001), 040401.

\bibitem{sz2011}
Sz. Szalay,
\it Separability criteria for mixed three-qubit states,
\rm Phys. Rev. A {\bf 83} (2011), 062337.

\bibitem{sz2012}
S. Szalay and Z. K\" ok\' enyesi,
\it Partial separability revisited: Necessary and sufficient criteria,
\rm Phys. Rev. A {\bf 86}, 032341 (2012).

\bibitem{sz2015}
Sz. Szalay,
\it Multipartite entanglement measures,
\rm Phys. Rev. A {\bf 92} (2015), 042329.

\bibitem{sz2018}
Sz. Szalay,
\it The classification of multipartite quantum correlation,
\rm J. Phys. A: Math. Theor. {\bf 51} (2018), 485302 .

\bibitem{han_kye_bisep_exam}
K. H. Han and S.-H. Kye,
\it Construction of three-qubit biseparable states distinguishing kinds of entanglement in a partial separability classification,
\rm Phys. Rev. A, {\rm 99} (2019), 032304.

\bibitem{han_kye_pe}
K. H. Han and S.-H, Kye,
\it On the convex cones arising from classifications of partial entanglement in the three qubit system,
\rm J. Phys. A: Math. Theor. {\bf 53} (2020), 015301.

\bibitem{Szalay-2019}
Sz. {\relax Sz}alay,
\it $k$-stretchability of entanglement, and the duality of $k$-separability and $k$-producibility,
\rm Quantum {\bf 3} (2019), 204.

\bibitem{han_kye_szalay}
K. H. Han, S.-H. Kye and S Szalay,
\it Partial separability/entanglement violates distributive rules,
\rm Quantum Inf. Process. {\bf 19} (2020), 202.

\bibitem{birkhoff}
G. Birkhoff,
Lattice Theory, 3rd ed.,
Amer. Math. Soc. Coll. Publ. Vol XXV, Amer. Math. Soc. 1967

\bibitem{freese}
R. Freese, J. Je\v zek and J. Nation,
Free Lattice, Math Surv. Mono. Vol 42, Amer. Math. Soc. 1991.

\bibitem{gao}
T. Gao and Y. Hong,
\it Separability criteria for several classes of $n$-partite quantum states,
\rm Eur. Phys. J. D {\bf 61} (2011), 765--771.

\bibitem{guhne10}
O. G\" uhne and M. Seevinck,
\it Separability criteria for genuine multiparticle entanglement,
\rm New J. Phys. {\bfone 2} (2010), 053002.

\bibitem{Rafsanjani}
S. M. H. Rafsanjani, M. Huber, C. J. Broadbent and J. H. Eberly
\it Genuinely multipartite concurrence of N-qubit X matrices,
\rm Phys. Rev. A {\bf 86} (2012), 062303.

\bibitem{han_kye_optimal}
K. H. Han and S.-H, Kye,
\it Construction of multi-qubit optimal genuine entanglement witnesses,
\rm J. Phys. A: Math. Theor. {\bf 49} (2016), 175303.

\bibitem{han_kye_tri}
K. H. Han and S.-H, Kye,
\it Various notions of positivity for bi-linear maps and applications to tri-partite entanglement,
\rm J. Math. Phys. {\bf 57} (2016), 015205.

\bibitem{yu07}
T. Yu and J. H. Eberly, \it Evolution from entanglement to
decoherence of bi-partite mixed \lq\lq X\rq\rq\  states, \rm Quantum
Inform. Comput. {\bf 7} (2007), 459--468.

\bibitem{GHZ}
D. M. Greenberger, M. A. Horne and A. Zeilinger,
\it Going beyond Bell's theorem,
\rm in  Bell's Theorem, Quantum Theory and Conceptions of the Universe,
Fundamental Theories of Physics, Vol. 37 (1989), 73--76.

\bibitem{han_kye_GHZ}
K. H. Han and S.-H, Kye,
\it Separability of three qubit Greenberger-Horne-Zeilinger diagonal states,
\rm J. Phys. A: Math. Theor. {\bf 50} (2017), 145303.


\end{thebibliography}
\end{document}